\documentclass[12pt]{article}
\usepackage{graphicx}
\usepackage{float}
\usepackage[flushleft]{caption2}
\usepackage{amsmath,amsfonts,amssymb}
\usepackage{amsthm}
\usepackage{mathrsfs}
\usepackage{enumitem}
\usepackage{tikz}
\usepackage[title]{appendix}
\usepackage{indentfirst}
\usepackage{epstopdf}
\usepackage{hyperref}
\usepackage{float}
\usepackage{pgfplots}
\usetikzlibrary{calc}  
\hypersetup{hidelinks}

\DeclareMathOperator*{\re}{Re}
\DeclareMathOperator*{\im}{Im}

\usepackage[sort&compress,numbers]{natbib}
\usetikzlibrary{arrows}
\newtheorem{lemma}{Lemma}[section]
\newtheorem{theorem}{Theorem}[section]

\newtheorem{corollary}{Corollary}[section]

\newtheorem{proposition}{Proposition}[section]

\setitemize[1]{itemsep=0pt,partopsep=0pt,parsep=\parskip,topsep=5pt}
\numberwithin{equation}{section}
\title{\LARGE\bf  {Asymptotic behavior analysis of solutions to the Heisenberg ferromagnet equation with the Schwartz initial data} }
\begin{document}
\author{\leftline{\hspace{0.6 cm}}\\ Xumeng Zhou$^{a}$, Xianguo Geng$^{b,c}$\footnote{\footnotesize
 Corresponding author. {\sl E-mail address}: xggeng@zzu.edu.cn}, Bo Xue$^{a}$\\
\leftline{\hspace{0.6 cm}{\small{$^{a}$ School of Mathematics and Statistics, Henan University, Kaifeng, Henan 475004, }}}\\
\leftline{\hspace{0.6 cm}{\small{\quad People's Republic of China}}}\\
\leftline{\hspace{0.6 cm}{\small{\sl $^{b}$ School of Mathematics and Statistics, North China University of Water Resources }}}\\
	\leftline{\hspace{0.6 cm}{\small{\quad and Electric Power, Zhengzhou, Henan 450011, People's Republic of China}}}\\
\leftline{\hspace{0.6 cm}{\small{\sl $^{c}$ School of Mathematics and Statistics, Zhengzhou University, 100 Kexue Road, Zhengzhou, }}}\\
	\leftline{\hspace{0.6 cm}{\small{\sl \quad Henan 450001, People's Republic of China}}}}

\date{}
\maketitle

\begin{abstract}
This work aims to investigate the long-time asymptotic behavior of solutions to the Cauchy problem for the Heisenberg ferromagnet equation with the Schwartz initial data. Utilizing the gauge transformations, spectral analysis and the inverse scattering method, we prove that the solutions to the Heisenberg ferromagnet equation can be expressed in terms of the solutions to a matrix Riemann-Hilbert problem formulated in the complex $k$-plane. Various Deift-Zhou contour deformations and the motivation behind them are given. By applying the nonlinear steepest descent method to the associated matrix-valued Riemann-Hilbert problem, we obtain the exact leading-order asymptotic formulas and uniform error estimates for solutions to the Cauchy problem of the Heisenberg ferromagnet equation.	\\ 

\end{abstract}

\section{Introduction}
	The Heisenberg ferromagnet (HF) equation is a fundamental integrable model used to describe the nonlinear evolution of one-dimensional spin waves. Its physical origin can be traced back to the quantum-mechanical exchange interaction model for ferromagnets proposed by Heisenberg in 1928 \cite{Heisenberg W}. 
This model characterizes the coupling interaction between adjacent spins and the motion of the magnetization vector of the isotropic ferromagnets. Later, it was derived in the form of nonlinear partial differential equations under the continuous medium approximation, and has become a classic research object in the interdisciplinary fields such as condensed matter physics, nonlinear optics, plasma physics, and integrable systems \cite{Takhtajan-Heisenberg,ZT-1979}. In the 1980s, Takhtajan extended the inverse scattering transform (IST) method to the continuous Heisenberg spin chain, providing a general framework for solving the equations of motion and giving an infinite number of conservation laws of the HF equation, thus laying the foundation for subsequent asymptotic analysis.
For this equation, the Darboux transformation was constructed, and its solition solutions and rogue waves were obtained \cite{Darboux-Heisenberg,Zha}. Demontis et al. derived a new general closed-form expression for the soliton solutions of the HF equation satisfying the in-plane asymptotic conditions via the inverse scattering transform and the matrix triple method \cite{CCHF}. The geometric relation associated with the HF model was studied by resorting to the motion of curves in Minkowski space \cite{Zha}.

The study of integrable systems underwent a revolutionary transformation with the advent of the IST \cite{IST-Gardner,Zakharov}. This method was established in the 1960s and provided a systematic approach for solving integrable nonlinear evolution equations. However, the complex analytical methods involved in IST have led to the development of a more rigorous mathematical reconstruction method, namely the Riemann-Hilbert (RH) problem. By recasting the scattering data as a boundary value problem in the complex plane, the RH method not only yields exact solutions for models such as the HF equation, but also provides a unified framework for the theory of integrable systems. Although the RH method has achieved success in constructing exact solutions, for a long time, its direct application to analyzing the long-time asymptotic behavior has been hindered by the notorious difficulty of estimating oscillatory integrals. In 1993, Deift and Zhou introduced the nonlinear steepest descent method in the RH problem, achieving a crucial breakthrough \cite{DeiftZ}. This pioneering framework transforms the originally intractable original RH problem into a series of solvable problems through three key steps: (1) Contour deformation, aligning the jump contour with the paths of steepest descent for the phase function; (2) Rational-oscillatory factorization of the jump matrix, to separate the dominant solitonic components; (3) Local rescaling near critical points, to resolve fine-scale behavior analytically. This method successfully decouples the soliton-dominated terms from the radiative decaying terms, providing a powerful and universal tool for the long-time asymptotic analysis of soliton equations. The universality of this method has enabled it to rapidly become an important tool for studying classical integrable equations associated with \(2\times2\) matrix spectral problems, such as the KdV equation, modified KdV (mKdV) equation, nonlinear Schr\"odinger (NLS) equation, sine-Gordon equation, Camassa-Holm equation and so on \cite{Andreiev-ELT2016,Minakov-2011,Liunan-mKdV,Biondini-M2017,Boutet-LS2021,Boutet-LS2022,Boutet-KS2022,Cheng-VZ1999,Deift-P2011,Kitaev-V1999}, and has systematically revealed the decay and oscillation behaviors of the solutions of these equations over long time scales. As the research progressed, the application scope of this method continued to expand, gradually extending to complex scenarios such as higher-order, multi-component and discrete systems. In the case of high-order matrix spectral problems, this method has been successfully extended to certain soliton equations associated with \(3\times3\) and even \(4\times4\) matrix spectral problems, such as the coupled nonlinear Schr\"odinger equation, the Sasa-Satsuma equation, the spin-1 Gross-Pitaevskii equation, etc \cite{Boutet-LS2019,Boutet-S2014,Boutet-SZ2016,Geng-L2018,Liu-GX2018,Geng-WC2021,Geng-WC2022}. Even when facing more complex RH problem structures and a larger number of stationary points, it can still accurately characterize the asymptotic behavior of solutions \cite{Hirota-F,coupled Hirota,TianSF,GuoBL,ZQMW-GI}. In the direction of discrete integrable systems, researchers have applied this method to models like the discrete focusing mKdV equation \cite{discrete-mKdV}, and have obtained the long-time asymptotic behavior of the solutions of this discrete integrable systems. 

In this paper, we study the long-time asymptotic behavior of solutions to the Cauchy problem for the HF equation:
\begin{equation}\label{HF}
\begin{cases}
    u_t= -\frac{1}{2}(u_xw-uw_x)_x,  \\
    v_t= \frac{1}{2}(v_xw-vw_x)_x,   \\
    u(x,0)=u_0(x), \quad v(x,0)=v_0(x),
\end{cases}
\end{equation}
where $u$ and $v$ are two potentials, and $w=\sqrt{1-uv}$. The initial values $u_0(x)$  and $v_0(x)$ lie in the Schwartz space
\begin{equation*}
  \mathscr{S}(\mathbb{R})=\{f(x)\in C^\infty(\mathbb{R}):\sup_{x\in\mathbb{R}}|x^\alpha\partial^\beta f(x)|<\infty,\forall\alpha,\beta\in\mathbb{N}\}.
\end{equation*}
Due to the existence of energy-dependent potentials and the constraint among potentials, the spectral analysis of the matrix eigenvalue problem is extremely difficult. The long-time asymptotic behavior of the solution to the Cauchy problem of the HF equation remains an unsolved challenge. The method we employ here is a combination of the inverse scattering transform and Deift-Zhou nonlinear steepest descent method. Resorting to the spectral analysis of Lax pair, the introduced transformations of field variables and transformations of independent variables, we first perform asymptotic analysis in the region $k \to \infty$ ($k$ is the spectral parameter) and then establish a local asymptotic expansion near $k=0$ to achieve the precise reconstruction of the potentials. Furthermore, due to the inherent lack of symmetry in the HF equation, we introduce an additional scattering matrix and complete the scattering relationship.
On this basis, the Deift-Zhou nonlinear steepest descent method is systematically applied to the constructed RH problem to rigorously derive the long-time asymptotic behavior of the solution, including the complete higher-order asymptotic correction terms. At this point, the long-time asymptotic behavior of the HF equation under the constraint of potentials has not yet been fully obtained. 

To address these critical gaps, the main contributions of this paper are summarized as follows: (1) We propose a novel two-step asymptotic analysis strategy ($k \to \infty$ followed by $k=0$ expansion) to solve the reconstruction problem of the potentials, that cannot be directly solved by traditional RH problems; (2) We introduce an auxiliary scattering matrix to address the deficiency of symmetry and provide a rigorous mathematical foundation for the subsequent asymptotic analysis; (3) We obtain a complete long-time asymptotic expansion that includes higher-order correction terms, systematically revealing the coupling mechanism in the radiation region; (4) We fill the research gap on the long-time asymptotic behavior of the HF equation under the constraint of potentials, thereby enriching and improving the asymptotic theory of multi-component integrable systems.

The main conclusions of this paper are as follows:
\begin{theorem}\label{the result}
    Let $u(x,t)$ and $v(x,t)$ be the solution to the Cauchy problem for the HF equation \eqref{HF} with initial data $u_0(x), v_0(x) \in \mathscr{S}(\mathbb{R})$. Assume that the spectral function $a(k)$ have no zeros in the upper half-line. As $t \to \infty$, the solution to the Cauchy problem for the HF equation \eqref{HF} admits the following asymptotic expression:

\begin{equation}
\begin{aligned}
  u(x,t)=& \frac{\mathrm{d}}{\mathrm{d}x}\left[ \frac{\alpha_u(\hat{k}_0)}{\sqrt{t}}e^{\Theta}+O\left(\frac{\log t}{t} \right) \right],  \\ 
  v(x,t)=& \frac{\mathrm{d}}{\mathrm{d}x}\left[ \frac{\alpha_v(\hat{k}_0)}{\sqrt{t}}e^{\tilde{\Theta}}+O\left(\frac{\log t}{t} \right) \right], 
  \end{aligned} 
\end{equation}
where
\begin{equation*}
\begin{aligned}
\alpha_u(\hat{k}_0)=&-\frac{1}{2\hat{k}_0^2}\sqrt{\frac{\nu \gamma_1(\hat{k}_0)}{i\gamma_2(\hat{k}_0)} }, \quad \alpha_v(\hat{k}_0)= \frac{1}{2\hat{k}_0^2}\sqrt{\frac{\nu \gamma_1(\hat{k}_0)}{i\gamma_2(\hat{k}_0)} },\\
  \Theta =& -\frac{\pi i}{4}+\pi \nu +2\hat{k}_0^2t- i\arg \Gamma (i \nu) -i\nu \log(4t)- i\arg \frac{\gamma_1(\hat{k}_0)}{1+\gamma_1(\hat{k}_0)\gamma_2(\hat{k}_0)} \\ 
& -\frac{1}{\pi}
  \int_{-\infty}^{\hat{k}_0} \log \left(\frac{1+\gamma_1(\xi)\gamma_2(\xi)}{1+\gamma_1(\hat{k}_0)\gamma_2(\hat{k}_0)}\right) \frac{\mathrm{d} \xi}{\xi -\hat{k}_0}, \\ 
  \tilde{\Theta} =& \frac{\pi i}{4}-\pi \nu -2\hat{k}_0^2t- i\arg \Gamma (-i \nu) +i\nu \log(4t)- i\arg \gamma_2(\hat{k}_0) \\ 
& +\frac{1}{\pi}
  \int_{-\infty}^{\hat{k}_0} \log \left(\frac{1+\gamma_1(\xi)\gamma_2(\xi)}{1+\gamma_1(\hat{k}_0)\gamma_2(\hat{k}_0)}\right) \frac{\mathrm{d} \xi}{\xi -\hat{k}_0},
\end{aligned}
\end{equation*}
and
\begin{equation*}
\begin{aligned}
  \hat{k}_0=&\frac{x}{2t}, \quad \nu=\frac{1}{2\pi} \log (1+\gamma_1(\hat{k}_0)\gamma_2(\hat{k}_0)).
\end{aligned}
\end{equation*}
\end{theorem}

The structure of this paper is arranged as follows:
In Section 2, we perform spectral analysis and construct the RH problem. We first introduce the Lax pair related to the HF equation, clarify the basic form of its \(2\times2\) matrix spectral problem, and then construct the Jost solutions. Since the spectral problem of this equation contains two potentials, we need to introduce two scattering coefficients to establish the associated RH problem. By analyzing the eigenfunctions near \(k=0\) and the correlation between the two eigenfunctions, we reconstruct the potentials. In Section 3 (long-time asymptotic analysis), based on the idea of the nonlinear steepest descent method, we decompose the jump matrix into the soliton-dominant part and the oscillatory part corresponding to radiation by deforming the contour, extending the contour of the RH problem, and performing scaling transformations. Finally, combining the boundedness analysis of the Cauchy operator and local scaling transformations, we derive the long-time asymptotic expansion of the solution, characterize the exponential decay laws in the soliton region and radiation region, and improve the long-time evolution behavior of the HF equation.

\section{Spectral analysis and Riemann-Hilbert problem}
Before we proceed to this section, we introduce the following notations.
\begin{itemize}
	 \item Let $\sigma_1, \sigma_2$ and $\sigma_3$ be the Pauli matrices, that is
\begin{equation*}
	\sigma_1= \begin{pmatrix} 0	&  1  \\ 1	& 0	\end{pmatrix}, \ \
	\sigma_2= \begin{pmatrix} 0	&  -i  \\ i	& 0	\end{pmatrix},  \ \
	\sigma_3= \begin{pmatrix} 1	&  0  \\ 0	& -1\end{pmatrix}.
\end{equation*}
Besides, in this article we define
\begin{equation*}
\sigma=-i \sigma_2=\begin{pmatrix} 0	&  -1  \\ 1	& 0	\end{pmatrix}.
\end{equation*}
    \item The operator $\hat{\sigma}_3$ acts on a $2 \times 2$ matrix $A$ by $\hat{\sigma}_3A=[\sigma_3, A]$, then $e^{\hat{\sigma}_3} A=
			e^{\sigma_3} A e^{-\sigma_3},$ where $[\cdot, \cdot]$ denotes the standard matrix commutator. Besides, for another $2 \times 2$ matrix $B$, we have $e^{\hat{A}}B=e^A B e^{-A}$.
    \item The superscript $``*"$ means the Schwartz conjugate, $``\dagger"$ means the Hermitian conjugate.
	\item For a $2 \times 2$ matrix $Z$, let $Z_1$ and $Z_2$ denote the first and second columns of $Z$, respectively. Let $Z^{(D)}$ denote the diagonal part of $Z$, and $Z^{(O)}$ denote the off-diagonal part of $Z$. Besides, let $Z_{ij}$ denote the $(i,j)$-entry of $Z$ for $i,j=1,2$.
	\item For two quantities $A$ and $B$, define $A\lesssim B$ if there exists a constant $C>0$ such that $A\leqslant CB$. 
\end{itemize}
\subsection{Transformation of the Lax pair}
The HF equation \eqref{HF} admits the following Lax pair:
	\begin{subequations}\label{Lax pair 1}
		\begin{align}
			\psi_x=U \psi, \\
			\psi_t=V \psi, 
		\end{align}
	\end{subequations}
where $\psi= (\psi_1, \psi_2)^T$ is a $2 \times 2$ matrix-valued function and
\begin{equation}
	U=\lambda \begin{pmatrix} w
		& u  \\  v  
		& -w
	\end{pmatrix}, \ \  \  \  
V=\lambda \begin{pmatrix} V_{11}
	& V_{12} \\  V_{21}
	& -V_{11}
\end{pmatrix}.
\end{equation}
Here $\lambda \in \mathbb{C}$ is the spectral parameter and 
\begin{equation}
	\begin{split}
		V_{11}&=-\frac{1}{4}  (uv_x-u_xv)- \lambda w, \\
		V_{12}&=-\frac{1}{4w} u(uv_x-u_xv)	-\frac{1}{2w} u_x-\lambda u, \\
		V_{21}&=-\frac{1}{4w}  v(uv_x-u_xv)	+\frac{1}{2w} v_x- \lambda v.
	\end{split}
\end{equation}

To simplify subsequent analysis, we set $\lambda=ik$ and introduce a gauge transformation
\begin{equation}\label{psi to psi tilde}
\psi(x,t,k)=G(x,t)\tilde{\psi}(x,t,k),
\end{equation}
where
\begin{equation}
	G=\sqrt{\frac{1+w}{2w}}  \begin{pmatrix} 1
		&  -\frac{u}{1+w}\\ \frac{w-1}{u}
		& 1
	\end{pmatrix}.
\end{equation}
As $x \to \pm \infty$, $G \to I$. In addition, the HF equation \eqref{HF} satisfies the following conservation law:
\begin{equation}
w_t=-\frac{1}{4}(uv_x-u_xv)_x.
\end{equation}
The Lax pair \eqref{Lax pair 1} can be rewritten as
\begin{subequations}\label{Lax pair 2}
\begin{align}
	\tilde{\psi}_x&=P_x \sigma_3\tilde{\psi}+\tilde{U}\tilde{\psi},\\
	\tilde{\psi}_t&=P_t \sigma_3 \tilde{\psi}+\tilde{V}\tilde{\psi},
\end{align}
\end{subequations}
where
\begin{equation}
P_x= ik w,\ \quad 	P_t= -\frac{1}{4} ik (uv_x-u_xv) + k^2, 
\end{equation}
\begin{equation}
\tilde{U}= \begin{pmatrix} \tilde{U}_{11}
		&  \tilde{U}_{12} \\ \tilde{U}_{21}
		&  -\tilde{U}_{11}
	\end{pmatrix}, \quad \tilde{V}= \begin{pmatrix} \tilde{V}_{11}
		&  \tilde{V}_{12} \\ \tilde{V}_{21}
		&  -\tilde{V}_{11}
	\end{pmatrix},
\end{equation}
and
\begin{equation*}
\begin{aligned}
\tilde{U}_{11} =& -\frac{u_xv+uv_x}{4w(1+w)}+ik(1-w), \\
\tilde{U}_{12} =& -\frac{uw_x-u_xw}{2w^2}, \\ 
\tilde{U}_{21} =& -\frac{2w_x+u_xv+uv_x}{4uw}, \\
\tilde{V}_{11} =& -\frac{u_tv+uv_t}{4w(1+w)}-\frac{ik(w+1)(uv_x-u_xv)-4k^2uv }{4w}+\frac{1}{4} ik  (uv_x-u_xv), \\ 
\tilde{V}_{12} =& -\frac{uw_t-u_tw}{2w^2}+\frac{ik u(uv_x-u_xv)}{4w}-k^2 u+\frac{ik u^2v(uv_x-u_xv)-2ik u^2v_x+4k^2u^2vw}{8w^2(1+w)} \\ 
              &-\frac{1+w}{8w^2}\left[ik u (1+w)(uv_x-u_xv) + 2ik u_x(1+w)  - 4k^2 u w(1+w) \right], \\  
\tilde{V}_{21} =& -\frac{2w_t+u_tv+uv_t}{4uw}+\frac{2ik v (1-w)-ik u^2 v (uv_x-u_xv)}{4w}+ k^2 u^2 v  + \frac{ik (1-w)u_x v}{4uw^2} \\ 
&+\frac{ik v (1-w)(uv_x-u_xv)}{8w^2}+\frac{1+w}{8w^2}\left[-ik v (uv_x-u_xv)+2ik v_x +4k^2 v\right],
\end{aligned}
\end{equation*}
with $\tilde{U}, \tilde{V} \to 0$, the matrix-valued function $\tilde{\psi}$ tends to a exponential function as $x \to \pm \infty$. Define
\begin{equation}
	P(x,t,k)=ik\left( x -\int_{x}^{+\infty} (w-1) \mathrm{d} \xi \right)+ k^2 t.
\end{equation}
Setting
\begin{equation} \label{psi to Phi}
\Phi (x,t,k) = \tilde{\psi}(x,t,k) e^{-P \sigma_3},
\end{equation} 
we obtain a modified Lax pair
\begin{subequations}\label{Lax pair 3}
\begin{align}
	\Phi_x&= P_x[\sigma_3,\Phi]+\tilde{U}\Phi, \label{Phi x} \\
	\Phi_t&=P_t[\sigma_3,\Phi]+\tilde{V}\Phi.
\end{align}
\end{subequations}

\subsection{Jost solutions}

Based on the Volterra integral equation, we introduce the Jost solutions $\Phi_{\pm}(x,t,k)$:
\begin{equation}
	\Phi_{\pm}(x,t,k)=I+\int_{\pm \infty}^{x} e^{(P(x)-P(\xi))\hat{\sigma}_3} \tilde{U}(\xi,t,k) \Phi_{\pm}(\xi,t ,k) \mathrm{d} \xi.
\end{equation}
\begin{figure}[H]
\vspace{0cm}
\centering
\begin{tikzpicture}
\draw[thick](4,0)--(0,0);
\draw[thick,<-](4,0)--(8,0);
\draw(8.5,-0.25) node [above] {$\mathbb{R}$};
\draw(4,0.5) node [above] {$\mathbb{C}_+$};
\draw(4,-1.3) node [above] {$\mathbb{C}_-$};
\end{tikzpicture}
\caption{The oriented contour on $\mathbb{R}$.}
\label{fig:1} 
\end{figure}
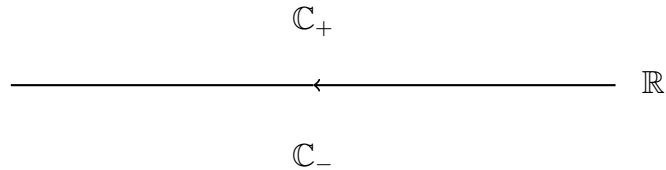
 \noindent The Jost solutions $\Phi_{\pm}=\left(\Phi_{\pm1},\Phi_{\pm2}\right)$ have the following properties (see figure \ref{fig:1}):
\begin{itemize}
\item $\Phi_{+1}$ and $\Phi_{-2}$ are analytic in $\mathbb{C}_+$;
\item $\Phi_{-1}$ and $\Phi_{+2}$ are analytic in $\mathbb{C}_-$;
\item $\Phi_{\pm}(x,t,k)=I+O(\frac{1}{k})$, \quad $k \to \infty$;
\item $\Phi_{\pm}(x,t,k)$ satisfy the symmetry relations $\overline{\Phi(-k)}=-\Phi(\bar{k}), \quad \Phi^{\dagger}(\bar{k})=\sigma \Phi^{-1}(k)\sigma$;
\item $\det \Phi_{\pm}(x,t,k)=1$.
\end{itemize}
Note that the two solutions $\Phi_{\pm}$ of differential equations \eqref{Lax pair 3} are related by a matrix independent of $x$ and $t$, that can be expressed as
\begin{equation}\label{S relate}
\Phi_-(x,t,k)=	\Phi_+(x,t,k)e^{P\hat{\sigma}_3}S(k), \quad k \in \mathbb{R}.
\end{equation}
From the symmetry properties of $\Phi_{\pm}(x,t,k)$, we find that $S(k)$ satisfies symmetric relations

\begin{equation*}
\overline{S(-k)}=-S(\bar{k}), \quad S^{\dagger}(\bar{k})=\sigma S^{-1}(k)\sigma.
\end{equation*}
Thus, $S(k)$ takes the form
\begin{equation*}
S(k)= \begin{pmatrix} a(k)
		& b(k) \\ c(k) & d(k)	\end{pmatrix}, \quad k \in \mathbb{R},
\end{equation*}
where $a(k)$ , $b(k)$, $c(k)$ and $d(k)$ satisfy
\begin{equation}\label{Symmetry abcd}
\begin{aligned}
&a^*(k)=-a(k)=a(-k), \quad b^*(k)=-b(k)=b(-k), \\
&c^*(k)=-c(k)=c(-k), \quad d^*(k)=-d(k)=d(-k), \quad k \in \mathbb{R}.
\end{aligned}
\end{equation}
Taking the determinant of both sides of \eqref{S relate} yields
\begin{equation}
\det S(k)=1.
\end{equation} 
Furthermore, the asymptotic properties of $\Phi_{\pm}(x,t,k)$ implies $a(k) \to 1$ as $k \to \infty$.
\subsection{Riemann-Hilbert Problem}
We define the matrix-valued function $\tilde{M} (x, t, k) $ by
\begin{equation}
	\tilde{M}(x,t,k)=\left\{
\begin{array}{l@{\quad}l} 
\left( \dfrac{\Phi_{-1}(x,t,k)}{a(k)}, \Phi_{+2}(x,t,k) \right), & k \in \mathbb{C}_-, \\
\left( \Phi_{+1}(x,t,k), \dfrac{\Phi_{-2}(x,t,k)}{d(k)} \right), & k \in \mathbb{C}_+.
\end{array}
\right.
\end{equation}
Introduce the following new scattering data: two reflection coefficients
\begin{equation}
	\gamma_1(k)=\frac{c(k)}{a(k)}, \quad \gamma_2=-\frac{b(k)}{d(k)},\quad k \in \mathbb{R}.
\end{equation}
Direct calculation indicates the following RH problem:
\begin{equation}\label{RHP M tilde}
	\tilde{M}(x,t,k)=\left\{
	\begin{array}{l@{\quad}l}  \tilde{M}_+(x,t,k)=\tilde{M}_-(x,t,k)J(x,t,k), & k \in \mathbb{R}, \\
		\tilde{M}(x,t,k) \to I, & k \to \infty,
	\end{array}
	\right.
\end{equation}
where
\begin{equation}\label{J x}
	J(x,t,k)=\begin{pmatrix} 1+\gamma_1 \gamma_2
		&  \gamma_2 e^{2P}\\ \gamma_1 e^{-2P}
		& 1
	\end{pmatrix}.
\end{equation}
To eliminate the soliton-type phenomena, we hereby assume that $\det a(k)$ and $\det  a^*(k)$ have no zeros. According to the properties of $\Phi_{\pm}(x,t,k)$ and $S(k)$, we immediately get the following lemma.
\begin{lemma} \label{lemma M tilde}
  The solution $\tilde{M}(x,t,k)$ to the RH problem \eqref{RHP M tilde} has the properties:
\begin{itemize}
	\item[(i)]  $\tilde{M} (x,t,k)$ is piecewise meromorphic in $\mathbb{C}\setminus\mathbb{R}$, with continuous boundary values $\tilde{M}_\pm(x,t,k)$ on $\mathbb{R}$. Furthermore, $\tilde{M}_+(x,t,k)$ is analytic in the upper complex $k$-plane  $\mathbb{C}_+$, and $\tilde{M}_-(x,t,k)$ is analytic in the lower complex $k$-plane $\mathbb{C}_-$;
	\item[(ii)] The normalization condition: $\tilde{M} (x, t,k) \to I $ as $ k \to \infty$;
	\item[(iii)] The symmetry conditions:
			$	\tilde{M}^{\dagger}(k^*)=\tilde{M}^{-1}(k), \quad \tilde{M}^*(-k)=\tilde{M}(k^*)$;
	\item[(iv)] $\det \tilde{M}_{\pm}(x,t,k)=1$.
\end{itemize}
\end{lemma}
Note that $\tilde{M}(x,t,k)$ admits the following asymptotic expansion as $k \to \infty$,
\begin{equation}
	\tilde{M}(x,t,k)=I+\frac{\tilde{M}_1(x,t)}{k}+\frac{\tilde{M}_2(x,t)}{k^2}+O(\frac{1}{k^3}).
\end{equation}
Substituting this expansion into \eqref{Phi x}, we have
\begin{equation}
	\lim_{k \to \infty}[k \tilde{M}(x,t,k)]_{12}=	\frac{uw_x-u_xw}{2iw^2(1+w)}, \quad  \lim_{k \to \infty}[k \tilde{M}(x,t,k)]_{21}=
	-\frac{2w_x+u_xv+uv_x}{4iuw(1+w)}.
\end{equation}
To this step, the potentials $u(x, t)$ and $v(x,t)$ can generally be reconstructed from the solutions of the fundamental RH problem. However, as shown in the above equation, reconstructing these potential functions is challenging and requires further analysis. Thus, we will concentrate on the transformation with the eigenfunctions near $k=0$.

\subsection{Eigenfunctions near $k=0$ and a new scale}
As $x \to \pm \infty, \ \psi(x,t,k) \sim  e^{ik\sigma_3 x + k^2 \sigma_3 t}$. To this end, we introduce the substitution
\begin{equation}\label{psi to psi hat}
	\hat{\psi}(x,t,k)=\psi(x,t,k) e^{-ik\sigma_3 x - k^2\sigma_3 t},
\end{equation}
via which the original Lax pair is transformed into
\begin{equation}
	\begin{split}
			\hat{\psi}_x&=ik[\sigma_3,\hat{\psi}]+\hat{U}_0 \hat{\psi},\\
			\hat{\psi}_t&=k^2[\sigma_3,\hat{\psi}]+\hat{V}_0 \hat{\psi},
		\end{split}
\end{equation}
where
\begin{equation}
	\hat{U}_0=ik \begin{pmatrix} w-1
		& u  \\  v  
		& 1-w
	\end{pmatrix},  \quad \hat{V}_0= \begin{pmatrix} V_{11}-k^2
	& V_{12}  \\  V_{21}  
	& -V_{11}+k^2
\end{pmatrix}.
\end{equation}
Based on this, we further introduce two additional eigenfunctions $\hat{\psi}_{\pm}(x,t,k)$, which satisfy the Volterra integrable equations as below
\begin{equation}
	\hat{\psi}_{\pm}(x,t,k)=I+\int_{\pm \infty}^{x}e^{ik (x-\xi)\hat{\sigma}_3}\hat{U}_0(\xi,t,k)\hat{\psi}_{\pm}(\xi,t,k)\mathrm{d}\xi.
\end{equation}
\begin{proposition}\label{Asymptotic proposition}
The functions $\hat{\psi}_{\pm}(x,t,k) $ have the expansions as follows:
\begin{equation}\label{hat psi pm expansion}
	\hat{\psi}_{\pm}(x,t,k)=I+ik \begin{pmatrix} \int_{\pm \infty}^{x}(w-1) \mathrm{d} \xi
		& \int_{\pm \infty}^{x} u \mathrm{d} \xi \\ \int_{\pm \infty}^{x} v  \mathrm{d} \xi
		&\int_{\pm \infty}^{x} (1-w) \mathrm{d} \xi
	\end{pmatrix} +O(k^2), \quad  k \to 0.
\end{equation}
\end{proposition}
From the asymptotic expansions \eqref{hat psi pm expansion} of $\hat{\psi}_{\pm}(x,t,k) $ near	$k = 0$, we can directly extract the leading-order information about the potentials $u$ and $v$ in the HF model. With this connection between the small-$k$ asymptotics and the potential reconstruction in hand, we now turn to establishing the linear relationship between $\hat{\psi}_{\pm}(x,t,k) $ and the eigenfunction functions $\Phi_{\pm}(x,t,k)$ defined in the preceding sections. 

\subsection{The linear relationship between two eigenfunction functions}
A linear relationship exists between the eigenfunctions $\Phi_{\pm}(x,t,k)$ and the eigenfunctions $\hat{\psi}_{\pm}(x,t,k)$. Both of them are derived from the eigenfunction transformations of $\psi(x,t,k)$ as specified in equations \eqref{psi to psi tilde}, \eqref{psi to Phi} and \eqref{psi to psi hat}. Consequently, there exist matrix-valued functions $\Gamma_{\pm}(k) $ (independent of $x$ and $t$) such that
\begin{equation}\label{Gamma pm}
	\Phi_{\pm}(x,t,k)=G^{-1}(x,t)\hat{\psi}_{\pm}(x,t,k)e^{ik\sigma_3 x+k^2\sigma_3t} \Gamma_{\pm}(k)e^{-P\sigma_3}.
\end{equation}
In the limit as $x \to \pm \infty$, we observe the asymptotic behavior $\Phi_{\pm}(x,t,k) \to I,$  $G(x,t) \to I$ and $ \hat{\psi}_{\pm}(x,t,k) \to I$. Substituting these limits into the above equation yields
\begin{equation}
	\begin{split}
			\Gamma_+ (k)& =I, \\
			 \Gamma_-(k) &=e^{-ik \int_{-\infty}^{+\infty} (w-1) \mathrm{d} \xi \sigma_3} \triangleq e^{-ik \tilde{c} \sigma_3}.
	\end{split}
\end{equation}
By the conservation law, $\tilde{c}$ is a constant independent of $x$ and $t$. Let $\tilde{c}=\tilde{c}_+(x,t) + \tilde{c}_-(x,t),$ where
\begin{equation}
	\begin{split}
	\tilde{c}_+(x,t)&=\int_{x}^{+\infty} (w(\xi)-1) \mathrm{d} \xi, \\
	\tilde{c}_-(x,t)&=\int_{-\infty}^{x} (w(\xi)-1) \mathrm{d} \xi.
		\end{split}
\end{equation}
Hence,
\begin{equation}
	\begin{split}
		\Phi_+(x,t,k)&=G^{-1}(x,t)\hat{\psi}_+ (x,t,k) e^{ik\tilde{c}_+ \sigma_3}, \\
		\Phi_-(x,t,k)&=G^{-1}(x,t)\hat{\psi}_- (x,t,k)e^{-ik\tilde{c}_- \sigma_3}.
		\end{split}
\end{equation}
Combining \eqref{hat psi pm expansion} with \eqref{Gamma pm}, we derive the expansion of the power series in $k$ as $k \to 0$,
\begin{equation}\label{Phi expansion}
	\begin{split}
		\Phi_+(x,t,k)&=G^{-1}\left[I-ik \begin{pmatrix} \int_{x}^{+\infty}(w-1) \mathrm{d} \xi
			& \int_{x}^{+\infty} u \mathrm{d} \xi \\ \int_{x}^{+\infty} v  \mathrm{d} \xi
			&\int_{x}^{+\infty} (1-w) \mathrm{d} \xi
		\end{pmatrix} +O(k^2)\right]e^{ik\tilde{c}_+ \sigma_3}, \\
		\Phi_-(x,t,k)&=G^{-1}\left[I+ik \begin{pmatrix} \int_{-\infty}^{x}(w-1) \mathrm{d} \xi
			& \int_{-\infty}^{x} u \mathrm{d} \xi \\ \int_{-\infty}^{x} v  \mathrm{d} \xi
			&\int_{-\infty}^{x} (1-w) \mathrm{d} \xi
		\end{pmatrix} +O(k^2)\right]e^{-ik\tilde{c}_- \sigma_3}.
	\end{split}
\end{equation}
Recollecting the relation between $\Phi_+(x,t,k)$ and $\Phi_-(x,t,k)$ given in \eqref{S relate}, we have $ a(k)=\det(\Phi_{+1}, \Phi_{-2})$. Furthermore, substituting the expansions \eqref{Phi expansion} leads to
\begin{equation}
a(k)=1 +O(k^2), \quad k \to 0.
\end{equation}
Starting from the definition of $\tilde{M}(x,t,k)$, we now reconstruct the potentials as $k \to 0$ via the following expression:
\begin{equation}
\tilde{M}(x,t,k)=\tilde{M}(x,t,0) \left[ I - ik\begin{pmatrix} \int_{x}^{+\infty}(w-1) \mathrm{d} \xi
			& \int_{x}^{+\infty} u \mathrm{d} \xi \\ \int_{x}^{+\infty} v  \mathrm{d} \xi
			&\int_{x}^{+\infty} (1-w) \mathrm{d} \xi
		\end{pmatrix} +O(k^2)  \right].
\end{equation} 
Therefore, we have
\begin{equation}
\tilde{u}(x,t)=\int_{x}^{+\infty} u(\xi,t) \mathrm{d} \xi=\lim_{k \to 0} \frac{\left[\tilde{M}(x,t,0)^{-1}\tilde{M}(x,t,k)\right]_{12}}{-ik},
\end{equation}
\begin{equation}
\tilde{v}(x,t)=\int_{x}^{+\infty} v(\xi,t) \mathrm{d} \xi=\lim_{k \to 0} \frac{\left[\tilde{M}(x,t,0)^{-1}\tilde{M}(x,t,k)\right]_{21}}{-ik}.
\end{equation}
It is shown that the matrix-valued function $\tilde{M}(x,t,k)$ contains all the necessary information for reconstructing the solutions of the initial value problem \eqref{HF} based on a matrix-valued RH problem. However, we find that the form of the elements $\gamma_1(k) e^{2P}$ and $\gamma_2(k) e^{-2P}$ in the jump matrix \eqref{J x} is quite complex, which increases the difficulty of analyzing the oscillation term later.
For convenience, we introduce a new variable
\begin{equation}
iy(x,t)=x-\tilde{c}_+(x,t),
\end{equation}
in terms of which the phase function becomes
\begin{equation}
P(y(x,t),t,k)=-k y(x,t) + k^2 t.
\end{equation}
 We now consider the RH problem \eqref{RHP M tilde} in the $ (y,t) $ scale. Defining $\tilde{M}(x,t,k)=M(y(x,t),t,k)$, we obtain the basic RH problem:
\begin{equation}\label{RHP M}
	\left\{\begin{array}{ll} M_+(y,t,k)=M_-(y,t,k)J(y,t,k), & \quad k \in \mathbb{R},
	\\ M(y,t,k) \to I, &  \quad k \to \infty,
	\end{array}\right.
	\end{equation}
where the jump matrix $J(y,t,k)$ and phase $\theta(y,t,k)$ are given by
\begin{align}
	J(y,t,k)&= \begin{pmatrix}  1+\gamma_1(k) \gamma_2(k)
	& \gamma_2(k) e^{2it \theta} \\ \gamma_1(k) e^{-2it\theta}
	& 1
	\end{pmatrix}, \label{J y} \\
	\theta(y,t,k) &=ik \frac{y}{t} - ik^2. \label{theta}
\end{align}
From the solution of the basic RH problem, the solution of the Cauchy problem to the HF equation can be derived by expanding $M(x,t,k)$ as $k \to 0$,
\begin{equation}
\tilde{u}(x,t)=\tilde{u}(y(x,t),t), \quad \tilde{v}(x,t)=\tilde{v}(y(x,t),t),
\end{equation}
where
\begin{equation}\label{tilde u}
\tilde{u}(y,t)=\lim_{k \to 0} \frac{\left[M(y,t,0)^{-1}M(y,t,k)\right]_{12}}{-ik},
\end{equation}
\begin{equation}\label{tilde v}
\tilde{v}(y,t)=\lim_{k \to 0} \frac{\left[M(y,t,0)^{-1}M(y,t,k)\right]_{21}}{-ik}.
\end{equation}

\section{Long-time asymptotics}	
With the basic RH problem established in the previous section, our current focus now shifts to the rigorous analysis of the long-time asymptotic behavior of the solution. A critical first step is to determine the stationary points of the phase function $\theta(k)$.
Noting that
\begin{equation}
   \frac{\mathrm{d} \theta (k)}{\mathrm{d} k } =i\frac{y}{t}-2ik,
\end{equation}
we obtain a stationary point $k_0$ from $\frac{\mathrm{d}\theta(k)}{\mathrm{d}k}=0$, where
\begin{equation}
k_0=\frac{\xi}{2}, \quad \xi=\frac{y}{t}.
\end{equation}
Then $\theta(k)$ can be written as 
\begin{equation}
\theta(k)=i(-k^2+2k_0k).
\end{equation}

\subsection{Reorientation}
To prepare for the nonlinear steepest descent analysis, we first decompose the jump matrix into upper/lower triangular factors. For $k > k_0$,  the jump matrix admits the factorization
\begin{equation}
	J(y,t,k)= \begin{pmatrix} 1
	&\gamma_2e^{2it\theta} \\ 0 &1
	\end{pmatrix} \begin{pmatrix} 1
	&0 \\ \gamma_1e^{-2it\theta} &1
	\end{pmatrix},
\end{equation}
where the first matrix can be analytically extended to the region $\re(i \theta)<0$, and the second matrix to $\re(i \theta)>0$. This decomposition is essential for establishing the exponential decay of the jump matrix as $t\to \infty$. For $k<k_0$, the situation is the opposite. We note another decomposition of the jump matrix
\begin{equation}
	J(y,t,k)=\begin{pmatrix} 1
	&0 \\ \frac{\gamma_1}{1+\gamma_1\gamma_2}e^{-2it\theta} &1
	\end{pmatrix} \begin{pmatrix} 1+\gamma_1\gamma_2
	&0 \\ 0  &\frac{1}{1+\gamma_1\gamma_2}  \end{pmatrix}
	\begin{pmatrix} 1
	&\frac{\gamma_2}{1+\gamma_1\gamma_2}e^{2it\theta} \\0  &1  \end{pmatrix}.
\end{equation}
To eliminate the middle diagonal matrix, we consider
\begin{equation}\label{M(1)}
	M^{(1)}(y,t,k)=M(y,t,k)(\delta(k))^{-\sigma_3},
\end{equation}
which transforms the original RH problem into
\begin{equation}
	M_+^{(1)}(y,t,k)=M_-^{(1)}(y,t,k)J^{(1)}(y,t,k),
\end{equation}
where $J^{(1)}(y,t,k)=(\delta_-)^{\sigma_3}J(y,t,k)(\delta_+(k))^{-\sigma_3}$. For $k>k_0$,
\begin{equation}
	J^{(1)}(y,t,k)=
	\begin{pmatrix}1	& \gamma_2 \delta_+^2e^{2it\theta} \\ 0 	& 1	\end{pmatrix}
	\begin{pmatrix}1	&0\\  \gamma_1 \delta_-^{-2} e^{-2it\theta}	& 1	\end{pmatrix}.
\end{equation}
For $k<k_0$, we have
\begin{equation}
	\begin{aligned}
		J^{(1)}(y,t,k) = &\begin{pmatrix}1 & 0 \\ \frac{\gamma_1}{1+\gamma_1 \gamma_2} \delta_-^{-2}  e^{-2it\theta}& 0 \end{pmatrix}
		\begin{pmatrix}(1+\gamma_1 \gamma_2 )\delta_-\delta_+^{-1}&0 \\ 0 & \frac{1}{1+\gamma_1 \gamma_2} \delta_-^{-1}\delta_+ \end{pmatrix}
		 \begin{pmatrix}1&\frac{\gamma_2}{1+\gamma_1 \gamma_2}\delta_+^2  e^{2it\theta}\\ 0 & 1 \end{pmatrix}.
	\end{aligned}
\end{equation}
We now introduce a scalar RH problem for $\delta(k)$:
\begin{equation}\label{RHP delta}
	\left\{\begin{array} {ll}\delta_+(k)=\delta_-(k)(1+\gamma_1(k) \gamma_2(k)), \quad & k<k_0, \\  \delta_+(k)=\delta_-(k),\quad  & k>k_0, \\
	\delta(k) \to 1, \quad & k \to \infty.
	\end{array}\right.
\end{equation}
Combining the symmetry property \eqref{Symmetry abcd} of scattering data, we conclude that the function $1+\gamma_1(k)\gamma_2(k)$ is real-valued and bounded for $k \in \mathbb{R}$. According to the Plemelj formula and the Vanishing Lemma \cite{VanishingLemma}, this scalar RH problem \eqref{RHP delta} admits a unique bounded solution:
\begin{equation}
	\delta(k)= (  k-k_0 )^{- i \nu} e^{\chi(k)},
\end{equation}
where
\begin{align}
	\nu &=\frac{1}{2 \pi} \log(1+\gamma_1(k_0)\gamma_2(k_0)),\label{nu} \\
	\chi(k)&=\frac{1}{2 \pi i}\int_{-\infty}^{k_0}\log \left ( 1+\gamma_1(\xi)\gamma_2(\xi) \right )  \frac{\mathrm{d} \xi}{\xi - k}. \label{Chi}
\end{align}
	\begin{theorem}
The function $\delta(k)$ is uniformly bounded in the complex plane.
	\end{theorem}
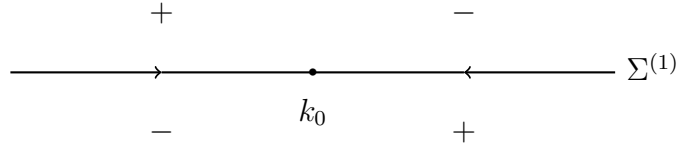
\begin{figure}[H]
\vspace{0cm}
\centering
\begin{tikzpicture}
\draw[thick,<-](2,0)--(0,0);
\draw[thick](2,0)--(4,0);
\draw[thick](4,0)--(6,0);
\draw[thick,<-](6,0)--(8,0);
\draw(4,0) node  {\tiny$\bullet$};
\draw(4,-0.2) node [below] {$k_0$};
\draw(8.5,-0.25) node [above] {$\Sigma^{(1)}$};
\draw(2,0.5) node [above] {$+$};
\draw(2,-0.5) node [below] {$-$};
\draw(6,0.5) node [above] {$-$};
\draw(6,-0.5) node [below] {$+$};
\end{tikzpicture}
\caption{The oriented contour on $\Sigma^{(1)}$.}
\label{fig:2} 
\end{figure}
	We reverse the orientation for the contour for $k \in (-\infty,k_0)$ as illustrated in figure \ref{fig:2}. This reorientation leads to a new RH problem for $M^{(1)}(y,t,k)$ on the oriented contour $\Sigma^{(1)}$:
	\begin{lemma}
		$M^{(1)}(y,t,k)$ is analytic on $\mathbb{C} \setminus \Sigma^{(1)} $ and satisfies the RH problem
		\begin{equation}\label{RHP M(1)}
		\left\{ \begin{array}{ll} M^{(1)}_+(y,t,k)=M^{(1)}_-(y,t,k) J^{(1)}(y,t,k), & \quad k \in \Sigma^{(1)}, \\
			M^{(1)}(y,t,k) \to I, & \quad k \to \infty,	\end{array} \right.
	\end{equation}
	where
	\begin{equation*}
	J^{(1)}(y,t,k)= \left\{\begin{array}{ll}  \begin{pmatrix} 1 & \gamma_2(k)e^{2it\theta}\delta_+^2(k) \\ 0	& 1 \end{pmatrix}   \begin{pmatrix} 1 & 0 \\ \gamma_1(k)e^{-2it\theta}	\delta_-^{-2}(k)& 1 \end{pmatrix} , \quad & k>k_0,
	\\ \begin{pmatrix} 1 & -\frac{\gamma_2(k)}{1+\gamma_1(k)\gamma_2(k)}e^{2it\theta}\delta_+^2(k) \\ 0	& 1 \end{pmatrix} \begin{pmatrix} 1 & 0 \\ -\frac{\gamma_1(k)}{1+\gamma_1(k)\gamma_2(k)}e^{-2it\theta}\delta_-^{-2}(k)	& 1 \end{pmatrix}
	, \quad & k<k_0.
	\end{array}\right.
\end{equation*}
	\end{lemma}
For unified processing, we introduce the representations
\begin{equation}
	\rho_1(k)= \left\{\begin{array}{ll}
		\gamma_1(k), \quad & k>k_0,	\\
		-\frac{\gamma_1(k)}{1+\gamma_1(k)\gamma_2(k)} , \quad & k<k_0,
	\end{array}\right.
\end{equation}
\begin{equation}
	\rho_2(k)= \left\{\begin{array}{ll}
		\gamma_2(k), \quad & k>k_0,	\\
		-\frac{\gamma_2(k)}{1+\gamma_1(k)\gamma_2(k)} , \quad & k<k_0,
	\end{array}\right.
\end{equation}
which allow us to unify the jump matrix on $\mathbb{R}$ as
\begin{equation}\label{J(1)}
	J^{(1)}(y,t,k)=\begin{pmatrix} 1		& -\rho_2(k) e^{2it\theta}  \delta_+^2(k )\\	0	&1	\end{pmatrix}^{-1} {\begin{pmatrix} 1	& 0 \\ \rho_1(k)e^{-2it\theta}  \delta_-^{-2}(k)		& 1	\end{pmatrix}}.
	\end{equation}
	We rewrite the expansions of $M(y,t,k)$ and $M^{(1)}(y,t,k)$ for $k \to 0$ as below
\begin{equation}
	M(y,t,k)=M_0(y,t)+M_1(y,t) k+O(k^2),
\end{equation}
\begin{equation}
	M^{(1)}(y,t,k)=M_0^{(1)}(y,t)+M_1^{(1)}(y,t)k+O(k^2).
\end{equation}
From \eqref{M(1)}, we have
\begin{equation}
	\begin{aligned}
	M(y,t,k)&=M^{(1)}(y,t,k)(\delta(k))^{\sigma_3} \\
	& = \left (  M_0^{(1)}(y,t)+M_1^{(1)}(y,t)k+O(k^2) \right )  \left (  I+\delta_1 \sigma_3 k+O(k^2) \right ) \\ & =  M_0^{(1)}(y,t)+   \left (  M_0^{(1)}(y,t) \delta_1 \sigma_3 + M_1^{(1)}(y,t) \right )k +O(k^2),
	\end{aligned}
\end{equation}
which means that
\begin{equation}
	M_0(y,t)=M_0^{(1)}(y,t), \quad  M_1(y,t)=M_0^{(1)}(y,t)\delta_1 \sigma_3 +M_1^{(1)}(y,t).
\end{equation}
Substituting these into \eqref{tilde u} and \eqref{tilde v}, we obtain the equivalent expressions 
\begin{equation}\label{tilde u2}
	\tilde{u}(y,t)=-\left(M_0^{-1}(y,t) M_1(y,t)\right)_{12} = -\left[\left( M_0^{(1)}(y,t)\right)^{-1}M_1^{(1)}(y,t)\right]_{12},
\end{equation}
\begin{equation}\label{tilde v2}
	\tilde{v}(y,t)=-\left(M_0^{-1}(y,t) M_1(y,t)\right)_{21} = -\left[\left( M_0^{(1)}(y,t)\right)^{-1}M_1^{(1)}(y,t)\right]_{21}.
\end{equation}
\subsection{The ``open lense" transformation }
In what follows, we shall split the spectral function $\rho_1(k)$ and establish a RH problem on an augmented contour $\Sigma^{(2)}=\mathbb{R} \cup L \cup L^*$, which is equivalent to the RH problem \eqref{RHP M(1)}, where
\begin{equation}\label{L definition}
	\begin{aligned}
	&L=\left\{k=k_0+k_0 \alpha e^{\frac{ 3\pi i}{4}}: -\infty < \alpha < +\infty\right\}, \\
	&L^{*}=\left\{k=k_0+k_0 \alpha e^{\frac{ \pi i}{4}}: -\infty < \alpha < +\infty\right\}.
	\end{aligned}
	\end{equation}
	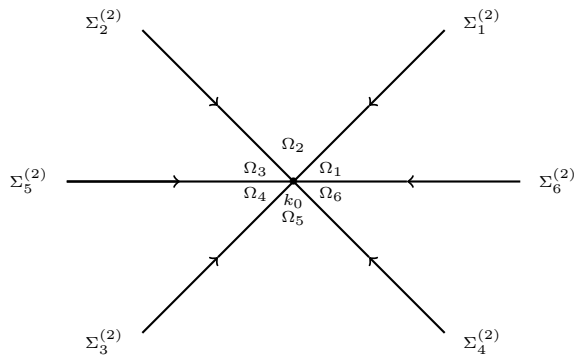
\begin{figure}[H]
\vspace{0cm}
\centering
\begin{tikzpicture}
\draw[thick](0,0)--(3,0);
\draw[thick,->](0,0)--(1.5,0);
\draw[thick,->](5,2)--(4,1);
\draw[thick](4,1)--(3,0);
\draw[thick](3,0)--(4.5,0);
\draw[thick,<-](4.5,0)--(6,0);

\draw[thick,->](5,-2)--(4,-1);
\draw[thick](4,-1)--(3,0);

\draw[thick,->](1,2)--(2,1);
\draw[thick](2,1)--(3,0);

\draw[thick,->](1,-2)--(2,-1);
\draw[thick](2,-1)--(3,0);
\draw(3,0) node  {\tiny$\bullet$};
\draw(3,-0.5) node [above] {\tiny$k_0$};
\draw(5.5,1.8) node [above] {\tiny$\Sigma_1^{(2)}$};
\draw(0.5,1.8) node [above] {\tiny$\Sigma_2^{(2)}$};
\draw(0.5,-1.8) node [below] {\tiny$\Sigma_3^{(2)}$};
\draw(5.5,-1.8) node [below] {\tiny$\Sigma_4^{(2)}$};
\draw(-0.5,-0.3) node [above] {\tiny$\Sigma_5^{(2)}$};
\draw(6.5,-0.3) node [above] {\tiny$\Sigma_6^{(2)}$};
\draw(3.5,0.4) node [below] {\tiny$\Omega_1$};
\draw(3.5,-0.4) node [above] {\tiny$\Omega_6$};
\draw(2.5,0.4) node [below] {\tiny$\Omega_3$};
\draw(2.5,-0.4) node [above] {\tiny$\Omega_4$};
\draw(3,0.25) node [above] {\tiny$\Omega_2$};
\draw(3,-0.25) node [below] {\tiny$\Omega_5$};
\end{tikzpicture}
\caption{The oriented contour on $\Sigma^{(2)}$.}
\label{fig:3} 
\end{figure}
\begin{theorem} \label{decomposition rho}The matrix-valued function $\rho_1(k)$ has a decomposition
\begin{equation}\label{rho decompose}
	\rho_1(k)=h_I(k)+h_{II}(k)+R(k), \quad  k \in \mathbb{R},
\end{equation}
where $h_I(k)$ cannot be analytically extended to $L$, but decays on $\mathbb{R}$, $h_{II}(k)$ is analytically and continuously extended to $L$, and $R(k)$ is a piecewise-rational function. In addition, for an arbitrary positive integer $k$, the following estimates is valid
\begin{align}
	& |e^{-2it\theta(k)}h_I(k)| \lesssim \frac{1}{(1+|k|^2)t^k}, \quad k \in \mathbb{R}, \label{h1}\\
	& |\frac{e^{-2it\theta(k)}h_I(k)}{k^j}| \lesssim  \frac{1}{(1+|k|^2)t^k}, \quad |k|>0, \quad j=1,2, \label{h1_j} \\
	& |e^{-2it\theta(k)}h_{II}(k)| \lesssim \frac{1}{(1+|k|^2)t^k}, \quad k \in L, \label{h2 L}\\
	& |\frac{e^{-2it\theta(k)}h_{II}(k)}{k^j}| \lesssim  \frac{1}{(1+|k|^2)t^k}, \quad k \in L, \quad j=1,2, \label{h2_j L}\\
	& |e^{-2it\theta(k)}R(k)| \lesssim e^{-2k_0^2 t}, \quad  k \in L. \label{R}
\end{align}
\end{theorem}
\begin{proof}
Assume $m=4d+1,$ $d$ is an integer and $d>2$. Expanding $(k-i)^{m+5} \rho(k)$ in a Taylor series around $k_0$, we arrive at
\begin{equation}\label{rho(lambda-i)}
	(k-i)^{m+5}\rho(k)=\mu_0+\mu_1(k-k_0)+\dots +\mu_m(k-k_0)^m+\frac{1}{m!}\int_{k_0}^{k}\left((\xi-i)^{m+5}\rho(\xi)\right)^{(m+1)}(k-\xi)^m\mathrm{d} \xi,
\end{equation}
where $\mu_0, \dots , \mu_m$ are Taylor coefficients. Then we define
\begin{align}
	R(k)&=\frac{\sum_{j=0}^{m}\mu_j(k-k_0)^j}{(k-i)^{m+5}},\\
	h(k)&=\rho(k)-R(k). \label{h}
\end{align}
Suppose
\begin{equation}
	\beta(k)=\frac{(k-k_0)^d}{(k-i)^{d+2}}. \label{beta}
\end{equation}
Applying the Fourier transform to $\frac{h}{\beta}$, we obtain
\begin{equation}
	\left(\frac{h}{\beta}\right)(k)=\int_{-\infty}^{+\infty}e^{is\theta}\widehat{	\left(\frac{h}{\beta}\right)}(s) \frac{\mathrm{d} s}{\sqrt{2\pi}} ,
\end{equation}
where the Fourier transform is defined by
\begin{equation}
	\widehat{	\left(\frac{h}{\beta}\right)}(s)=\int_{k_0}^{+\infty} e^{-is\theta}	\left(\frac{h}{\beta}\right)(k)\frac{\mathrm{d} \theta(k)}{\sqrt{2\pi}}.
\end{equation}
From \eqref{rho(lambda-i)}, \eqref{h} and \eqref{beta}, it follows directly that
\begin{equation}
	\left(\frac{h}{\beta}\right)(k)=\frac{(k-k_0)^{3d+2}}{(k-i)^{3d+4}}g(k,k_0),
\end{equation}
with
\begin{equation}
	g(k,k_0)=\frac{1}{m!}\int_{0}^{1}\left[\left(k_0+\zeta(k-k_0)-i \right)^{m+5}\rho(k_0+\zeta(k-k_0))\right] ^{(m+1)}(1-\zeta)^m \mathrm{d}\zeta.
\end{equation}
Noting the uniform bound
\begin{equation}
	\left| \frac{\mathrm{d}^j g(k,k_0)}{\mathrm{d} k^j}   \right| \lesssim 1, \quad    k \ge k_0,
\end{equation}
we derive the estimates
\begin{equation}\label{h beta}
	\left \| \left(\frac{\mathrm{d} }{\mathrm{d} \theta } \right)^j \left(\frac{h}{\beta}\right)(k(\theta))\right \| ^2_{\mathscr{L}^2(\mathbb{R})} \lesssim 1, \quad   0 \le j \le \frac{3d+2}{3}.
\end{equation}
We further split $h(k)$ into two components:
\begin{equation}
	\begin{aligned}
	h(k)&=\beta(k)\int_{t}^{+\infty}e^{is\theta(k)}\widehat{\left(\frac{h}{\beta}\right)}(s)\frac{\mathrm{d} s}{\sqrt[]{2\pi } } +\beta(k)\int_{-\infty}^{t}e^{is\theta(k)}\widehat{\left(\frac{h}{\beta}\right)}(s)\frac{\mathrm{d} s}{\sqrt[]{2\pi } }\\ & \triangleq h_I(k)+h_{II}(k). 	
	\end{aligned}
\end{equation}
From \eqref{h beta}, we conclude that $h/\beta \in H^j(-\infty,+\infty)$ for $ 0 \le j \le \frac{3d+2}{3}$, which yields the following inequality:

\begin{align}
	&\int_{-\infty}^{+\infty}(1+s^2)^j \left|\widehat{\left(\frac{h}{\beta}\right)}(s) \right|^2 \mathrm{d} s \lesssim 1, \\
	&\int_{-\infty}^{+\infty}s^j \left|\widehat{\left(\frac{h}{\beta}\right)}(s) \right|^2 \mathrm{d} s \lesssim 1.
\end{align}
For $k \ge k_0$, we derive the following estimate for $h_I(k)$:
\begin{equation}
	\begin{aligned}\label{eh1}
	\left|  e^{-2it\theta} h_I(k)\right| & \le  |\beta(k)|  \int_{t}^{+\infty} \left| \widehat{\left(\frac{h}{\beta}\right)}(s) \right| \frac{\mathrm{d} s}{\sqrt[]{2\pi} } \\
	& \le |\beta(k)|\left(\int_{t}^{\infty}(1+s^2)^{-p} \frac{\mathrm{d} s}{\sqrt{2\pi} } \right)^{\frac{1}{2}}
	\left(\int_{t}^{\infty}(1+s^2)^{p} \left| \widehat{\left(\frac{h}{\beta}\right)}(s) \right |^2 \frac{\mathrm{d} s}{\sqrt{2\pi} } \right)^{\frac{1}{2}} \\
	& \lesssim \frac{t^{\frac{1}{2}-p}}{1+|k|^2}.
	\end{aligned}
\end{equation}
Next, for $j=1,2,$
	\begin{equation}
		\begin{aligned}
			\left| \frac{e^{-2it \theta}h_I(k)}{k^j} \right| &=
			\left| \frac{(k-k_0)^d}{k^j(k-i)^{d+2}} \right|\left| \int_{t}^{+\infty} e^{is \theta}
			\widehat{	\left(\frac{h}{\beta}\right)}(s) \frac{\mathrm{d}s}{\sqrt{2 \pi }} \right| \\
			&\lesssim \frac{1}{(k-i)^{2+j}} \left \|  \widehat{	\left(\frac{h}{\beta}\right)}(s) \right \| _{\mathscr{L}^2(\mathbb{R})} \\
			& \le \frac{1}{(k-i)^{2+j}}  \left(\int_{t}^{+\infty} s^{-2r} ds\right)^{\frac{1}{2}}
			\left \|  s^r \widehat{	\left(\frac{h}{\beta}\right)}(s) \right \|_{\mathscr{L}^2(\mathbb{R})}\\
			& \lesssim \frac{1}{1+|k|^2}t^{-r+\frac{1}{2}}, \quad  r\le \frac{3d+2}{3}.
		\end{aligned}
	\end{equation}
Turning to $h_{II}(k)$, the condition $t-s>0$ allows us to decompose the phase factor as
\begin{equation}
e^{i(s-t)\theta}=e^{(s-t)\re(i \theta)}e^{i(s-t)\im(i \theta)}.
\end{equation}
This representation implies that $h_{II}(k)$ admits an analytic continuation to the line 
$L $ defined in \eqref{L definition}. On this line, we have
	\begin{equation}
		\re(i\theta)= k_0^2 >0.
	\end{equation}
Subsequently, we obtain the decay estimate for $h_{II}(k)$:
	\begin{equation}
	\begin{aligned}
		|e^{-2it\theta }h_{II}(k)| & \le e^{-t \re (i \theta)} \frac{|k-k_0|^d}{|k-i|^{d+2}}
		\left|\int_{-\infty}^{t}e^{(s-t)\re (i\theta)} \widehat{\left(\frac{h}{\beta}\right)}(s) \frac{\mathrm{d} s}{\sqrt{2\pi } } \right| \\
		& \lesssim 	\frac{k_0 ^d \alpha^d e^{-t \re(i \theta)}}{|k-i|^{d+2}} \left(\int_{-\infty}^{t} (1+s^2)^{-1}\right)^{\frac{1}{2}}\left(\int_{-\infty}^{t}(1+s^2)\left|\widehat{\left(\frac{h}{\beta}\right)}(s)\right| ^2 \right ) ^{\frac{1}{2}} \\ 
		& \lesssim \frac{t^{-\frac{d}{2}}}{1+|k|^2}.
	\end{aligned}
	\end{equation}
Similarly, for $j=1,2,$
	\begin{equation}
		\begin{aligned}
			\left| \frac{e^{-2it \theta}h_{II}(k)}{k^j} \right| &=
			e^{-2t \re(i\theta)}	\left| \frac{(k-k_0)^d}{k^j(k-i)^{d+2}} \right|
			\left| \int_{-\infty}^{t} e^{is \theta}
			\widehat{	\left(\frac{h}{\beta}\right)}(s) \frac{\mathrm{d}s}{\sqrt{2 \pi }} \right| \\
			&\lesssim		e^{-t \re(i\theta)} \frac{(\alpha k_0)^{d-j}}{|k-i|^{2}} \left| \int_{-\infty}^{t} e^{(s-t)\re( i\theta)}
			\widehat{	\left(\frac{h}{\beta}\right)}(s) \frac{\mathrm{d}s}{\sqrt{2 \pi }} \right| \\
			&\lesssim e^{-t \re(i\theta)} \frac{(\alpha k_0)^{d-j}}{|k-i|^{2}} \left(\int_{-\infty}^{t}
			(1+s^2)^{-1} ds\right)^{\frac{1}{2}}
			\left \|  (1+s^2) \widehat{	\left(\frac{h}{\beta}\right)}(s) \right \|_{\mathscr{L}^2(\mathbb{R})} \\
			& \lesssim e^{-t \re(i\theta)	} \frac{(\alpha k_0)^{d-j}}{|k-i|^{2}} \\
			& \lesssim \frac{1}{(1+|k|^2)}t^{\frac{j-d}{2}}.
		\end{aligned}
	\end{equation}
We now choose $ m = 4d + 1$ sufficiently large such that $\min \left \{ \frac{d}{2}, \ p-\frac{1}{2}, \  \frac{3d+2}{3}-\frac{1}{2}, \ \frac{d-j}{2}\right \}>k.$ Then estimates \eqref{h1}, \eqref{h1_j}, \eqref{h2 L} and \eqref{h2_j L} follow immediately. For $k \in L$, on the line $\left\{k=k_0+k_0 \alpha e^{-\frac{\pi i}{4}} : -\infty < \alpha < +\infty\right\}$, simple calculations indicate that
\begin{equation}
	\re (i\theta) = k_0^2.
	\end{equation}
Combining this with the boundedness of $R(k)$, we further establish the exponential decay estimate for $R(k)$:
\begin{equation}
	\left|e^{-2it\theta} R(k)\right|\lesssim \left|e^{-2it\theta} \right| = e^{-2k_0^2 t} \le e^{-\epsilon^2 t}, \quad \sqrt{2} k_0 \ge \epsilon.
\end{equation}
 This completes the proof of estimate \eqref{R}.
	\end{proof}
Similarly, $\rho_2(k)$ has the following decomposition:
\begin{equation}\label{rho2 decompose}
	\rho_2(k)=h_I^{\prime}(k)+h_{II}^{\prime}(k)+R^{\prime}(k), \quad  k \in \mathbb{R}.
\end{equation}
  The analysis for $\rho_2(k)$ is entirely analogous to that for $\rho_1(k)$. Thus, $e^{-2it\theta}h_I^{\prime}(k),$  $e^{-2it\theta}h_{II}^{\prime}(k)$ and $e^{-2it\theta}R^{\prime}(k)$ satisfy the same decay properties as their counterparts for $\rho_1(k)$.
\par Next, we will transform the oscillating RH problem \eqref{RHP M(1)} to a standard form via factorization. We decompose the jump matrix \eqref{J(1)} as
	\begin{equation*}
		J^{(1)}(y,t,k)=b_-^{-1}(k)b_+(k),
	\end{equation*}
	where the factors $b_{\pm}(k)$ are defined by
		\begin{equation}
	\begin{aligned}
		b_- (k)& =b_-^o(k) b_-^a(k)=\left(I- \omega_-^o(k)\right)\left(I-\omega_-^a(k)\right) \\
		&=  \begin{pmatrix} 1 & h_{I,\rho_2}(k) e^{2it\theta}\delta_+^2(k) \\ 0		&1	\end{pmatrix}
		\begin{pmatrix} 1 & \left[h_{II,\rho_2}(k)+[\rho_2](k)\right] e^{2it\theta}\delta_+^2(k) \\ 0		&1	\end{pmatrix},
		\\
		b_+(k) & =b_+^o(k) b_+^a(k)=\left(I+ \omega_+^o(k)\right)\left(I+\omega_+^a(k)\right) \\
		&=  \begin{pmatrix} 1 & 0 \\ h_{I,\rho_1}(k) e^{-2it\theta}\delta_-^{-2}(k)		&1	\end{pmatrix}
		\begin{pmatrix} 1 & 0 \\ \left[h_{II,\rho_1}(k)+[\rho_1](k)\right] e^{-2it\theta}\delta_-^{-2}(k)	&1	\end{pmatrix}.
	\end{aligned}
	\end{equation}
	Hence, we rewrite $J^{(1)}(y,t,k)$ in the form
	\begin{equation*}
		J^{(1)}(y,t,k)=b_-^{-1}(k)b_+(k)=\underset{L^{*}}{\underbrace{ \left(b_-^a(k)\right)^{-1}}} \ \
		\underset{\mathbb{R}}{\underbrace{\left(b_-^o(k)\right)^{-1}b_+^o(k)}} \ \
		\underset{L}{\underbrace{b_+^a(k)}}.
		\end{equation*}
Note that in the above three parts, the left term $ \left(b_-^a(k)\right)^{-1} $ has an analytic continuation to $L^{*}$, and the intermediate term $ \left(b_-^o(k)\right)^{-1}b_+^o(k) $ cannot be analytically extended, but it decays rapidly on $\mathbb{R}$ with respect to the time variable $t$, the right term $  b_+^a (k)$ has an analytic continuation to $L $. So we introduce the following transformation,
\begin{equation}\label{M(2) relate M(1)}
	M^{(2)}(y,t,k)=M^{(1)}(y,t,k)\varphi(k),
\end{equation}
where
\begin{equation}\label{varphi}
	\varphi(k)= \left\{\begin{array}{ll} I , \quad & k \in \Omega_2 \cup \Omega_5,\\ \left(b_-^a (k)\right)^{-1}, \quad   & k \in \Omega_1 \cap \Omega_4, \\ \left(b_+^a (k)\right)^{-1}, \quad & k \in \Omega_3 \cap \Omega_6 .		
	\end{array}\right.
\end{equation}
This transformation maps the original RH problem on the contour $\Sigma^{(1)}= \mathbb{R}$ to a new RH problem on the augmented contour $\Sigma^{(2)}$.
\begin{lemma}
	The matrix-valued function $M^{(2)}(y,t,k)$ is analytic on $\mathbb{C} \setminus \Sigma^{(2)} $ and satisfies the RH problem
		\begin{equation}\label{RHP M(2)}
		\left\{ \begin{array}{ll} M_+^{(2)}(y,t,k)=M_-^{(2)}(y,t,k)J^{(2)}(y,t,k), & \quad k \in \Sigma^{(2)}, \\
			M^{(2)}(y,t,k) \to I, &\quad  k \to \infty,	\end{array} \right.
	\end{equation}
	where the new jump matrix $J^{(2)}(y,t,k)$ is given by
	\begin{equation}
		\begin{aligned}\label{J(2) expression}
			J^{(2)}(y,t,k)= \left\{	\begin{array}{ll}  \left(b_-^a(k)\right)^{-1} , & \quad k \in L^{*},\\ b_+^a(k), & \quad k \in L , \\ \left(b_-^o(k)\right)^{-1}b_+^o(k), & \quad k \in \mathbb{R}.    \end{array}    \right.
		\end{aligned}
	\end{equation}
	\end{lemma}
We further decompose the jump matrix as follows,
\begin{equation}
J^{(2)}=(b_-^{(2)})^{-1}b_+^{(2)},
\end{equation}
where
	\begin{equation*}
		b_+^{(2)} (k)= \left\{	\begin{array}{ll} I , & k \in L^{*},\\ b_+^a (k), & k \in L , \\ b_+^o (k), & k \in \mathbb{R},    \end{array}    \right. \  \  \
		b_-^{(2)} (k)= \left\{	\begin{array}{ll} b_-^a (k) , & k \in L,\\ I, & k \in L^{*} , \\ b_-^o (k), & k \in \mathbb{R},    \end{array}    \right.
	\end{equation*}
and introduce
\begin{equation}
w_{\pm}^{(2)}=\pm(b_{\pm}^{(2)}-I), \quad w^{(2)}=w_+^{(2)}+w_-^{(2)}.
\end{equation}

	Next, we define the Cauchy operators $C_\pm$ for $k \in \Sigma^{(2)}$ via the singular integral
\begin{equation}
	(C_\pm f)(k)= \frac{1}{2 \pi i}\int_{\Sigma^{(2)}} \frac{f(\xi)}{\xi - k_\pm} \mathrm{d} \xi ,  \quad  
	f \in \mathscr{L}^2 (\Sigma),
\end{equation}
where $C_+f\ (C_-f)$ denotes the left (right) boundary value of $f$ along the oriented contour $\Sigma^{(2)}$ (see figure \ref{fig:3}). Let $\omega^{(2)}(k)=\omega_+^{(2)}(k)+\omega_-^{(2)}(k)$. We then introduce the operator $C_{\omega^{(2)}}: \mathscr{L}^2(\Sigma^{(2)}) + \mathscr{L}^{\infty}(\Sigma^{(2)}) \to \mathscr{L}^2(\Sigma^{(2)}) $ by
\begin{equation}\label{C w(2)-operater}
	C_{\omega^{(2)}} f = C_+(f \omega_-^{(2)}) + C_-(f \omega_+^{(2)}),
\end{equation}
	for any $2\times 2$ matrix-valued function $f$. Assume $\mu^{(2)}(y,t,k) \in \mathscr{L}^2(\Sigma^{(2)}) + \mathscr{L}^{\infty}(\Sigma^{(2)}) $  is the solution to the singular integral equation
\begin{equation}
	\mu^{(2)}(y,t,k)=I+C_{\omega^{(2)}} \mu^{(2)}(y,t,k).
\end{equation}
Then, according to the Beals-Coifman theorem \cite{BC}, the solution to the RH problem \eqref{RHP M(2)} is given by the Cauchy integral
\begin{equation}\label{M(2)-Cauchy}
	M^{(2)}(y,t,k)=I+\frac{1}{2 \pi i} \int_{\Sigma^{(2)}} \frac{\mu^{(2)}(y,t,\xi)\omega^{(2)}(y,t,\xi)}{\xi -k}
	\mathrm{d} \xi .
\end{equation}
 Noticing that $\gamma_1(k)=O(k^3)$ as $k \to 0,$ we expand $M^{(2)}(y,t,k)$ in powers of $k$:
 \begin{equation}
 		M^{(2)}(y,t,k) =  M_0^{(2)}(y,t)+	M_1^{(2)}(y,t)k + O(k^2).
 \end{equation}
From \eqref{M(2) relate M(1)} and \eqref{varphi}, it can be obtained that
\begin{equation}
		M_0^{(1)}(y,t) = M_0^{(2)}(y,t), \quad 	M_1^{(1)}(y,t) =M_1^{(2)}(y,t).
\end{equation}
Consequently, \eqref{tilde u2} and \eqref{tilde v2} are equivalent to
\begin{align}
	&	\tilde{u}(y,t)=-\left[\left( M_0^{(2)}(y,t) \right)^{-1} M_1^{(2)}(y,t)\right]_{12}, \label{tilde u3}\\
	&	\tilde{v}(y,t)=- \left[\left(M_0^{(2)}(y,t)\right)^{-1} M_1^{(2)}(y,t)\right]_{21}, \label{tilde v3}
\end{align}
where $ M_0^{(2)} (y,t)$ and $M_1^{(2)} (y,t)$ can be derived from \eqref{M(2)-Cauchy},
\begin{align}
	&	M_0^{(2)}(y,t)=I+ \frac{1}{2 \pi i} \int_{\Sigma^{(2)}} \frac{\mu^{(2)}(y,t,\xi)\omega^{(2)}(y,t,\xi)}{\xi} \mathrm{d} \xi, \label{M 0 (2)}\\
	&	M_1^{(2)}(y,t)= \frac{1}{2 \pi i} \int_{\Sigma^{(2)}} \frac{\mu^{(2)}(y,t,\xi)\omega^{(2)}(y,t,\xi)}{\xi^2} \mathrm{d} \xi. \label{M 1 (2)}
\end{align}
\subsection{Rational approximation of the Riemann-Hilbert problem}

 In this section, we perform a rational approximation of the RH problem and further decompose it as follows,
 \begin{equation}
 w^{(2)}=\hat{w}+R, \quad R=R_++R_-,
 \end{equation}
 where $\hat{w}$ consists of the $h_I$ and $h_{II}$ components of $\rho_1$ and $\rho_2$, while $R$ represents the rational part of $\rho_1$ and $\rho_2$. We now explicitly perform this decomposition on each of the six jump contours $\Sigma_j^{(2)} \ (j=1,2,\cdots,6)$: \\
 On $\Sigma_1^{(2)}$:
\begin{align}
J^{(2)}&={\begin{pmatrix} 1	& -\delta_+^{2} e^{2it\theta} (h_{II,\gamma_2}+[\gamma_2]) \\ 0		& 1	\end{pmatrix}}^{-1}, \quad
w^{(2)}={\begin{pmatrix} 0	& \delta_+^{2} e^{2it\theta} (h_{II,\gamma_2}+[\gamma_2]) \\ 0		& 0	\end{pmatrix}}, \\ 
\hat{w}&= {\begin{pmatrix} 0	& \delta_+^{2} e^{2it\theta} h_{II,\gamma_2} \\ 0		& 0	\end{pmatrix}}, \quad 
R={\begin{pmatrix} 0	& \delta_+^{2} e^{2it\theta} [\gamma_2] \\  0		& 0	\end{pmatrix}}; \label{w hat 2}
\end{align}
 on $\Sigma_2^{(2)}$:
\begin{align}
J^{(2)}&={\begin{pmatrix} 1	& 0 \\ \delta_-^{-2} e^{-2it\theta} (h_{II,\frac{-\gamma_1}{1+\gamma_1\gamma_2}}+[\frac{-\gamma_1}{1+\gamma_1\gamma_2}])		& 1	\end{pmatrix}}, \\
w^{(2)}&={\begin{pmatrix} 0	& 0 \\ \delta_-^{-2} e^{-2it\theta} (h_{II,\frac{-\gamma_1}{1+\gamma_1\gamma_2}}+[\frac{-\gamma_1}{1+\gamma_1\gamma_2}])		& 0	\end{pmatrix}}, \\ 
\hat{w}&= {\begin{pmatrix} 0	& 0 \\ \delta_-^{-2} e^{-2it\theta} h_{II,\frac{-\gamma_1}{1+\gamma_1\gamma_2}}		& 0	\end{pmatrix}}, \\  \label{w hat 1}
R&={\begin{pmatrix} 0	& 0 \\ \delta_-^{-2} e^{-2it\theta} [\frac{-\gamma_1}{1+\gamma_1\gamma_2}]		& 0	\end{pmatrix}};
\end{align}
on $\Sigma_3^{(2)}$:
\begin{align}
J^{(2)}&={\begin{pmatrix} 1	& -\delta_+^{2} e^{2it\theta} (h_{II,\frac{-\gamma_2}{1+\gamma_1\gamma_2}}+[\frac{-\gamma_2}{1+\gamma_1\gamma_2}]) \\ 0		& 1	\end{pmatrix}}^{-1}, \\
w^{(2)}&={\begin{pmatrix} 0	& \delta_+^{2} e^{2it\theta} (h_{II,\frac{-\gamma_2}{1+\gamma_1\gamma_2}}+[\frac{-\gamma_2}{1+\gamma_1\gamma_2}]) \\ 0		& 0	\end{pmatrix}}, \\ 
\hat{w}&= {\begin{pmatrix} 0	& \delta_+^{2} e^{2it\theta} h_{II,\frac{-\gamma_2}{1+\gamma_1\gamma_2}} \\ 0		& 0	\end{pmatrix}}, \\ \label{w hat 4}
R&={\begin{pmatrix} 0	& \delta_+^{2} e^{2it\theta} [\frac{-\gamma_2}{1+\gamma_1\gamma_2}] \\  0		& 0	\end{pmatrix}};
\end{align}
on $\Sigma_4^{(2)}$:
\begin{align}
J^{(2)}&={\begin{pmatrix} 1	& 0 \\ \delta_-^{-2} e^{-2it\theta} (h_{II,\gamma_1}+[\gamma_1])		& 1	\end{pmatrix}}, \quad
w^{(2)}={\begin{pmatrix} 0	& 0 \\ \delta_-^{-2} e^{-2it\theta} (h_{II,\gamma_1}+[\gamma_1])		& 0	\end{pmatrix}}, \\ 
\hat{w}&= {\begin{pmatrix} 0	& 0 \\ \delta_-^{-2} e^{-2it\theta} h_{II,\gamma_1}		& 0	\end{pmatrix}}, \quad 
R={\begin{pmatrix} 0	& 0 \\ \delta_-^{-2} e^{-2it\theta} [\gamma_1]		& 0	\end{pmatrix}}; \label{w hat 3}
\end{align}
on $\Sigma_5^{(2)}$:
\begin{align}
J^{(2)}&= \begin{pmatrix} 1	& -\delta_+^{2} e^{2it\theta} h_{I,\frac{-\gamma_2}{1+\gamma_1\gamma_2}} \\ 0		& 1	\end{pmatrix} ^{-1} \begin{pmatrix} 1	& 0 \\ \delta_-^{-2} e^{-2it\theta} h_{I,\frac{-\gamma_1}{1+\gamma_1\gamma_2}}		& 1	\end{pmatrix} , \\ 
w^{(2)}&=\hat{w}=\begin{pmatrix} 0	& -\delta_+^{2} e^{2it\theta} h_{I,\frac{-\gamma_2}{1+\gamma_1\gamma_2}} \\ 0		& 0	\end{pmatrix} ^{-1} \begin{pmatrix} 0	& 0 \\ \delta_-^{-2} e^{-2it\theta} h_{I,\frac{-\gamma_1}{1+\gamma_1\gamma_2}}		& 0	\end{pmatrix} , \\   \label{w hat 6}
R&=0;
\end{align}
on $\Sigma_6^{(2)}$:
\begin{align}
J^{(2)}&={\begin{pmatrix} 1	& -\delta_+^{2} e^{2it\theta} h_{I,\gamma_2} \\ 0		& 1	\end{pmatrix}}^{-1} {\begin{pmatrix} 1	& 0 \\ \delta_-^{-2} e^{-2it\theta} h_{I,\gamma_1}		& 1	\end{pmatrix}}, \\ 
w^{(2)}&=\hat{w}={\begin{pmatrix} 0	& -\delta_+^{2} e^{2it\theta} h_{I,\gamma_2} \\ 0		& 0	\end{pmatrix}}^{-1} {\begin{pmatrix} 0	& 0 \\ \delta_-^{-2} e^{-2it\theta} h_{I,\gamma_1}		& 0	\end{pmatrix}}, \\  \label{w hat 5}
R&=0.
\end{align}
We then define the Cauchy operators:
\begin{equation}
C_R f =C_+(fR_-)+C_-(fR_+),  \quad C_{\hat{w}} f =C_+(f \hat{w}_-)+C_-(f\hat{w}_+).
\end{equation}
Next, we demonstrate that the contribution to the solution of the RH problem mainly comes from the rational part of the scatter data, while the contribution of $h_I$ and $h_{II}$ is the infinitesimal quantity in $t$.
\begin{equation}\label{mu w 2}
\begin{aligned}
\int_{\Sigma^{(2)}} \mu ^{(2)} w^{(2)} 
=&\int_{\Sigma^{(2)}} [(1-C_{w^{(2)}})^{-1}I]w^{(2)} \\ 
=& \int_{\Sigma^{(2)}} [(1-C_{w^{(2)}})^{-1}(1-C_{w^{(2)}}+C_{w^{(2)}})I]w^{(2)} \\ 
=& \int_{\Sigma^{(2)}} w^{(2)} + \int_{\Sigma^{(2)}} [(1-C_{w^{(2)}})^{-1}C_{w^{(2)}}I]w^{(2)} \\ 
=& \int_{\Sigma^{(2)}} w^{(2)} + \int_{\Sigma^{(2)}} \left \{(1-C_R)^{-1}(1-C_R)(1-C_{w^{(2)}})^{-1}C_{w^{(2)}}I \right \}w^{(2)}  \\ 
=& \int_{\Sigma^{(2)}} w^{(2)} + \int_{\Sigma^{(2)}} \left \{(1-C_R)^{-1}(1-C_{w^{(2)}}+C_{\hat{w}})(1-C_{w^{(2)}})^{-1}C_{w^{(2)}}I \right \}w^{(2)} \\ 
=& \int_{\Sigma^{(2)}} w^{(2)} + \int_{\Sigma^{(2)}} \left \{(1-C_R)^{-1}C_{w^{(2)}}I\right \}w^{(2)} \\
&+\int_{\Sigma^{(2)}} \left \{(1-C_R)^{-1}C_{\hat{w}} (1-C_{w^{(2)}})^{-1} C_{w^{(2)}} I\right \} w^{(2)}  \\ 
=&  \int_{\Sigma^{(2)}} R + \int_{\Sigma^{(2)}} \hat{w} + \int_{\Sigma^{(2)}} \left \{ (1-C_R)^{-1} C_{w^{(2)}} I\right\} w^{(2)} \\ 
&+ \int_{\Sigma^{(2)}} \left \{  (1-C_R)^{-1} C_{\hat{w}} (1-C_{w^{(2)}})^{-1} C_{w^{(2)}} I \right\} w^{(2)}.
\end{aligned}
\end{equation}
Note that the third integral in the above equation can be written as

\begin{equation}
\begin{aligned}
& \int_{\Sigma^{(2)}} \left \{ (1-C_R)^{-1} C_{w^{(2)}} I \right \} w^{(2)} \\
   =& \int_{\Sigma^{(2)}} \left \{ (1-C_R)^{-1} (C_{\hat{w}}+C_R ) I\right\} w^{(2)} \\ 
   =& \int_{\Sigma^{(2)}} \left \{ (1-C_R)^{-1} C_{\hat{w}} I\right\} w^{(2)} +
   \int_{\Sigma^{(2)}} \left \{ (1-C_R)^{-1} C_R  I\right\} \hat{w}     \\
   &+\int_{\Sigma^{(2)}} \left \{ (1-C_R)^{-1} C_R  I\right\} R    \\ 
   =&  \int_{\Sigma^{(2)}} \left \{ (1-C_R)^{-1} C_{\hat{w}} I\right\} w^{(2)} +
    \int_{\Sigma^{(2)}} \left \{ (1-C_R)^{-1} C_R  I\right\} \hat{w}     \\
   &+\int_{\Sigma^{(2)}} \left \{ (1-C_R)^{-1}(1-(1-C_R) )  I\right\} R    \\ 
   =&  \int_{\Sigma^{(2)}} \left \{ (1-C_R)^{-1} C_{\hat{w}} I\right\} w^{(2)} + 
      \int_{\Sigma^{(2)}} \left \{ (1-C_R)^{-1} C_R  I\right\} \hat{w}     \\ 
     &+\int_{\Sigma^{(2)}} \left \{ (1-C_R)^{-1})  I\right\} R  -  \int_{\Sigma^{(2)}} R.      
\end{aligned}
\end{equation}
Substituting the above equation into \eqref{mu w 2}, we obtain
\begin{equation}\label{mu2 w2}
\begin{aligned}
& \int_{\Sigma^{(2)}} \mu ^{(2)} w^{(2)} \\
=& \int_{\Sigma^{(2)}} \left \{ (1-C_R)^{-1})  I\right\} R + \int_{\Sigma^{(2)}} \hat{w} + \int_{\Sigma^{(2)}} \left \{ (1-C_R)^{-1} C_{w^{(2)}} I\right\} w^{(2)} \\ 
&+ \int_{\Sigma^{(2)}} \left \{ (1-C_R)^{-1} C_R  I\right\} \hat{w} +\int_{\Sigma^{(2)}} \left \{  (1-C_R)^{-1} C_{\hat{w}} (1-C_{w^{(2)}})^{-1} C_{w^{(2)}} I \right\} w^{(2)} \\ 
\triangleq & \int_{\Sigma^{(2)}} \left \{ (1-C_R)^{-1})  I\right\} R + \mathcal{I} + \mathcal{II} + \mathcal{III} + \mathcal{IV}.
\end{aligned}
\end{equation}
Before estimating the contributions of the four terms $\mathcal{I}, \mathcal{II}, \mathcal{III} $ and $\mathcal{IV}$ to the solution of the RH problem, we first give the following result.
\begin{proposition}\label{CR}
For $|k_0|<M$ and as $t \to \infty,$ $(1-C_R)^{-1}$ exists and is uniformly bounded:
\begin{equation}
\parallel (1-C_R)^{-1} \parallel_{\mathscr{L}^2(\Sigma^{(2)}) \to \mathscr{L}^2(\Sigma^{(2)})} \lesssim 1.
\end{equation}
\begin{proof}
See the Proposition 2.23 and Corollary 2.25 in \cite{DeiftZ}.
\end{proof}
\end{proposition}
Based on Proposition \ref{CR}, we immediately derive the following conclusion.
\begin{corollary}
For $|k_0|<M$ and as $t \to \infty,$ $(1-C_{w^{(2)}})^{-1}$ exists and is uniformly bounded:
\begin{equation}
\parallel (1-C_{w^{(2)}})^{-1} \parallel_{\mathscr{L}^2(\Sigma^{(2)}) \to \mathscr{L}^2(\Sigma^{(2)})} \lesssim 1.
\end{equation}
\end{corollary}
\begin{proof} Since $C_{\pm}$ are bounded operators of $ \mathscr{L}^2(\Sigma) \to \mathscr{L}^2(\Sigma) $, we have
\begin{equation*}
	\begin{aligned}
	\left \| C_{\omega^{(2)}} -C_R \right \| _{\mathscr{L}^2(\Sigma) \to \mathscr{L}^2(\Sigma)} & =\sup \frac{\left \| C_+(f \hat{w}_-) + C_-(f \hat{w}_+) \right \|_ {\mathscr{L}^2(\Sigma)}}{\left \| f \right \|_ {\mathscr{L}^2(\Sigma)}} \\
	& \le \sup \frac{\left \| C_+f \right \| _ {\mathscr{L}^2(\Sigma)} \left \| \hat{w}_- \right \| _ {\mathscr{L}^{\infty}(\Sigma)} }{\left \|   f \right \| _ {\mathscr{L}^2(\Sigma)}} +
	\sup \frac{\left \| C_- f \right \| _ {\mathscr{L}^2(\Sigma)} \left \| \hat{w}_+ \right \| _ {\mathscr{L}^{\infty}(\Sigma)} }{\left \| f \right \| _ {\mathscr{L}^2(\Sigma)}}\\
	& \le  \left \| C_+ \right \| _ {\mathscr{L}^2(\Sigma) \to \mathscr{L}^2(\Sigma) } \left \| \hat{w}_- \right \| _ {\mathscr{L}^{\infty}(\Sigma)}
	+ \left \| C_- \right \| _ {\mathscr{L}^2(\Sigma) \to \mathscr{L}^2(\Sigma) } \left \| \hat{w}_+ \right \| _ {\mathscr{L}^{\infty}(\Sigma)} \\
	& \lesssim t^{-1}.
	\end{aligned}
\end{equation*}
Furthermore, the second resolvent identity gives
\begin{equation}
	(1-C_{\omega^{(2)}})^{-1}-(1-C_R)^{-1}=	(1-C_{\omega^{(2)}})^{-1}(C_{\omega^{(2)}}-C_R)(1-C_R)^{-1}.
\end{equation}
Combining this with Proposition \ref{CR}, we conclude that $	(1-C_{\omega^{(2)}})^{-1} $ is uniformly bounded.
\end{proof}
In what follows, we present the estimates for $\mathcal{I}, \mathcal{II}, \mathcal{III} $ and $\mathcal{IV}$. By observing matrices \eqref{w hat 1}, \eqref{w hat 2}, \eqref{w hat 3}, \eqref{w hat 4}, \eqref{w hat 5} and \eqref{w hat 6}, we see that the four matrix elements of $\hat{w}$ are composed of $h_{I,\rho_1}$, $h_{I,\rho_2}$, $h_{II,\rho_1}$ and $h_{II,\rho_2}$. Applying Proposition \ref{CR}, we derive the following results:
\begin{equation}\label{I estimate}
|\mathcal{I}| = \left | \int_{\Sigma^{(2)}} \hat{w}   \right |  \le \parallel  \hat{w} \parallel_{\mathscr{L}^1} \lesssim t^{-1},
\end{equation}
and
\begin{equation}
\begin{aligned}
|\mathcal{II}|=& \left | \int_{\Sigma^{(2)}} \left\{ (1-C_R)^{-1} C_{\hat{w}} I \right\} w^{(2)} \right | \\ 
\le & \parallel (1-C_R)^{-1} C_{\hat{w}} I \parallel_{\mathscr{L}^2} \parallel w^{(2)} \parallel_{\mathscr{L}^2} \\ 
\le & \parallel (1-C_R)^{-1}  \parallel_{\mathscr{L}^2 \to \mathscr{L}^2} 
\parallel C_{\hat{w}} I \parallel_{\mathscr{L}^2} \parallel w^{(2)} \parallel_{\mathscr{L}^2} \\ 
\lesssim & \parallel C_+ \hat{w}_- \parallel_{\mathscr{L}^2} + \parallel C_- \hat{w}_+ \parallel_{\mathscr{L}^2} \\ 
\lesssim & \parallel \hat{w} \parallel_{\mathscr{L}^2} \parallel w^{(2)} \parallel_{\mathscr{L}^2}, 
\end{aligned}
\end{equation}
where
\begin{equation}
\parallel w^{(2)} \parallel_{\mathscr{L}^2} \le \parallel \hat{w} \parallel_{\mathscr{L}^2}+
\parallel R \parallel_{\mathscr{L}^2},
\end{equation}
and
\begin{equation}
|R|=\left| \frac{\sum^m_{j=0} \mu_j (k-k_0)^j}{(k+i)^{m+5}} \right| \lesssim 
\frac{1}{|k+i|^5} \in \mathscr{L}^1 \cap \mathscr{L}^2 \cap \mathscr{L}^{\infty}.
\end{equation}
Thus, we have
\begin{equation}\label{II estimate}
|\mathcal{II}| \le ct^{-1} + ct^{-2} \lesssim t^{-1}.
\end{equation}
Similarly, we have the following estimates:
\begin{equation}\label{III estimate}
\begin{aligned}
|\mathcal{III}|
=& \left| \int_{\Sigma^{(2)}} \left \{ (1-C_R)^{-1} C_R  I\right\} \hat{w}  \right| \\ 
\le & \parallel (1-C_R)^{-1} \parallel_{\mathscr{L}^2 \to \mathscr{L}^2} 
\parallel C_R I \parallel_{\mathscr{L}^2} \parallel \hat{w} \parallel_{\mathscr{L}^2} \\ 
\le & c (\parallel C_+R_-\parallel_{\mathscr{L}^2} + \parallel C_-R_+\parallel_{\mathscr{L}^2} )
\parallel \hat{w} \parallel_{\mathscr{L}^2} \\ 
\le & c \parallel R \parallel_{\mathscr{L}^2} \parallel \hat{w} \parallel_{\mathscr{L}^2}  \\ 
\lesssim & t^{-1},
\end{aligned}
\end{equation}
and

\begin{equation}\label{IV estimate}
\begin{aligned}
|\mathcal{IV}|=& \left|\int_{\Sigma^{(2)}} \left \{  (1-C_R)^{-1} C_{\hat{w}} (1-C_{w^{(2)}})^{-1} C_{w^{(2)}} I \right\} w^{(2)} \right| \\ 
\le & \parallel (1-C_R)^{-1} \parallel_{\mathscr{L}^2 \to \mathscr{L}^2}  
\parallel C_{\hat{w}} \parallel_{\mathscr{L}^2 \to \mathscr{L}^2} 
\parallel (1-C_{w^{(2)}})^{-1} \parallel_{\mathscr{L}^2 \to \mathscr{L}^2} \\ 
& \cdot (\parallel C_+ w^{(2)}_- \parallel_{\mathscr{L}^2} + \parallel C_-w^{(2)}_+ \parallel_{\mathscr{L}^2}) \parallel w^{(2)} \parallel_{\mathscr{L}^2} \\ 
\lesssim & t^{-1}.
\end{aligned}
\end{equation}
Substituting \eqref{I estimate}, \eqref{II estimate}, \eqref{III estimate} and \eqref{IV estimate} into \eqref{mu w 2}, we have
\begin{equation}\label{w2 to R}
\int_{\Sigma^{(2)}} [(1-C_{w^{(2)}})^{-1}I]w^{(2)} = \int_{\Sigma^{(2)}} \left \{ (1-C_R)^{-1})  I\right\} R + O(t^{-1}).
\end{equation}
Setting $\mu ^{(3)}=(1-C_R)^{-1} I$ and $w^{(3)}=R$, we establish a new RH problem on the contour $\Sigma^{(3)}=\Sigma^{(2)} \setminus \mathbb{R}$  (see figure \ref{fig:4}). The corresponding solution is given by the Cauchy integral
\begin{equation}\label{M(3)-Cauchy}
	M^{(3)}(y,t,k)= I+ \frac{1}{2 \pi i}\int_{\Sigma^{(3)}}\frac{\mu^{(3)}(y,t,\xi) w^{(3)}(y,t,\xi)}{\xi -k} \mathrm{d} \xi.
\end{equation}

\begin{figure}[H]
\vspace{0cm}
\centering
\begin{tikzpicture}
\draw[thick,->](5,2)--(4,1);
\draw[thick](4,1)--(3,0);

\draw[thick,->](5,-2)--(4,-1);
\draw[thick](4,-1)--(3,0);

\draw[thick,->](1,2)--(2,1);
\draw[thick](2,1)--(3,0);

\draw[thick,->](1,-2)--(2,-1);
\draw[thick](2,-1)--(3,0);
\draw(3,0) node  {\tiny$\bullet$};
\draw(3,-0.6) node [above] {\tiny$k_0$};
\draw(5.5,1.8) node [above] {\tiny$\Sigma_1^{(3)}$};
\draw(0.5,1.8) node [above] {\tiny$\Sigma_2^{(3)}$};
\draw(0.5,-1.8) node [below] {\tiny$\Sigma_3^{(3)}$};
\draw(5.5,-1.8) node [below] {\tiny$\Sigma_4^{(3)}$};
\end{tikzpicture}
\caption{The oriented contour on $\Sigma^{(3)}$.}
\label{fig:4} 
\end{figure}
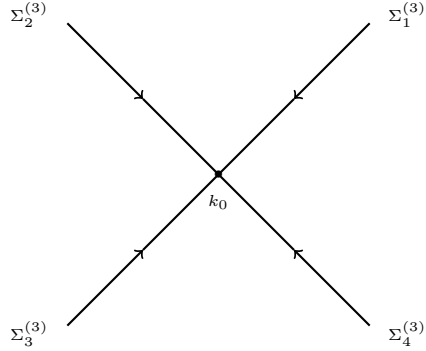
\begin{lemma}
	The matrix-valued function	$M^{(3)}(y,t,k)$ is analytic on $\mathbb{C} \setminus \Sigma^{(3)} $, and satisfies the RH problem
		\begin{equation}\label{RHP M(3)}
	\left\{\begin{array}{ll} M^{(3)}_+(y,t,k)=  M^{(3)}_-(y,t,k)  J^{(3)}(y,t,k) , \ & k \in \Sigma^{(3)},\\
	M^{(3)}(y,t,k) \to I, \ & k \to \infty,  \end{array}\right.
\end{equation}
where
\begin{equation}
\begin{aligned}
J^{(3)}(y,t,k)=\left\{\begin{array}{ll} 	
								\left( \begin{array}{cc} 1	& -\delta_+^{2} e^{2it\theta} [\gamma_2] \\  0		& 1	\end{array}\right)^{-1}, & k \in \Sigma^{(3)}_1, \\\left(\begin{array}{cc} 
 1	& 0 \\ \delta_-^{-2} e^{-2it\theta} [\frac{-\gamma_1}{1+\gamma_1\gamma_2}]		& 1	\end{array}\right), &  k \in \Sigma^{(3)}_2, \\ 
\left( \begin{array}{cc} 1	& -\delta_+^{2} e^{2it\theta} [\frac{-\gamma_2}{1+\gamma_1\gamma_2}] \\  0		& 1	\end{array}\right)^{-1}, &k \in \Sigma^{(3)}_3 , \\  
\left( \begin{array}{cc} 1	& 0 \\ \delta_-^{-2} e^{-2it\theta} [\gamma_1]		& 1	\end{array}\right), & k \in \Sigma^{(3)}_4.
\end{array}\right.
\end{aligned}
\end{equation}
\end{lemma}

We now expand $M^{(3)}(y,t,k)$ in powers of $k$ as $k \to 0$:
\begin{equation}
	M^{(3)}(y,t,k)=M_0^{(3)}(y,t)+M_1^{(3)}(y,t) k +O(k^2),  \quad  k \to 0.
\end{equation}
From the Cauchy integral representation \eqref{M(3)-Cauchy}, a direct calculation shows that
\begin{align}
	& M_0^{(3)}(y,t)= I+\frac{1}{2 \pi i} \int_{\Sigma^{(3)}}\frac{\left((1-C_{w^{(3)}})^{-1} I \right)(\xi) {w^{(3)}}(\xi)}{\xi } \mathrm{d} \xi, \label{M 0 (3)} \\
	& M_1^{(3)}(y,t)= \frac{1}{2 \pi i} \int_{\Sigma^{(3)}}\frac{\left((1-C_{w^{(3)}})^{-1} I\right)(\xi) w^{(3)}(\xi)}{\xi^2 } \mathrm{d} \xi. \label{M 1 (3)}
\end{align}
Combining \eqref{M 0 (2)}, \eqref{M 1 (2)}, \eqref{M 0 (3)}, \eqref{M 1 (3)} and \eqref{w2 to R}, we rewrite the potential reconstruction formulas \eqref{tilde u3} and \eqref{tilde v3} as
\begin{equation}
	\begin{aligned}\label{tilde u4}
	\tilde{u}(y,t )=&-\left[\left(M_0^{(2)}(y,t)\right)^{-1} M_1^{(2)}(y,t)\right]_{12} \\
	=&-\left\{ \left[M_0^{(3)}(y,t)+O( t^{-1})\right]^{-1} \left[M_1^{(3)}(y,t)+O( t^{-1})\right] \right\}_{12} \\
	=& -\left[M_0^{(3)}(y,t)+O( t^{-1})\right]_{22} \left[M_1^{(3)}(y,t)+O( t^{-1})\right]_{12} \\
	&+\left[M_0^{(3)}(y,t)+O( t^{-1})\right]_{12} \left[M_1^{(3)}(y,t)+O(t^{-1})\right]_{22} ,
	\end{aligned}
\end{equation}
\begin{equation}
	\begin{aligned}\label{tilde v4}
\tilde{v}(y,t)=&- \left[\left(M_0^{(2)}(y,t)\right)^{-1} M_1^{(2)}(y,t)\right]_{21} \\ 
=&-\left\{ \left[M_0^{(3)}(y,t)+O( t^{-1})\right]^{-1} \left[M_1^{(3)}(y,t)+O( t^{-1})\right] \right\}_{21} \\
	=& -\left[M_0^{(3)}(y,t)+O( t^{-1})\right]_{11} \left[M_1^{(3)}(y,t)+O( t^{-1})\right]_{21} \\
	&+\left[M_0^{(3)}(y,t)+O( t^{-1})\right]_{21} \left[M_1^{(3)}(y,t)+O(t^{-1})\right]_{11}.
	\end{aligned}
\end{equation}

\subsection{Scaling of Riemann-Hilbert problem}
In the next step, we perform a scaling transformation on the RH problem \eqref{RHP M(3)}. The leading-order asymptotics of $\tilde{u}(y,t)$ are determined by integrals on the contour $\Sigma^{(3)}$ by \eqref{tilde u4}. To obtain explicit expressions relating the potential $\tilde{u}(y,t)$ to the RH solution, we construct a scaled RH problem on a cross-shaped contour passing through the origin. We first extend the contour as follows:
\begin{align}
	\Sigma^{(3)}=\hat{\Sigma}=\left\{k_0 + \alpha k_0 e^{\pm \frac{\pi i}{4}} : \alpha \in \mathbb{R} \right\}.
\end{align}

\begin{figure}[H]
\vspace{0cm}
\centering
\begin{tikzpicture}
\draw[thick,->](5,2)--(4,1);
\draw[thick](4,1)--(3,0);

\draw[thick,->](5,-2)--(4,-1);
\draw[thick](4,-1)--(3,0);

\draw[thick,->](1,2)--(2,1);
\draw[thick](2,1)--(3,0);

\draw[thick,->](1,-2)--(2,-1);
\draw[thick](2,-1)--(3,0);
\draw(3,0) node  {\tiny$\bullet$};
\draw(3,-0.6) node [above] {$0$};
\draw(5.5,1.8) node [above] {$\Sigma_1$};
\draw(0.5,1.8) node [above] {$\Sigma_2$};
\draw(0.5,-1.8) node [below] {$\Sigma_3$};
\draw(5.5,-1.8) node [below] {$\Sigma_4$};
\end{tikzpicture}
\caption{The oriented contour on $\Sigma$.}
\label{fig:5} 
\end{figure}
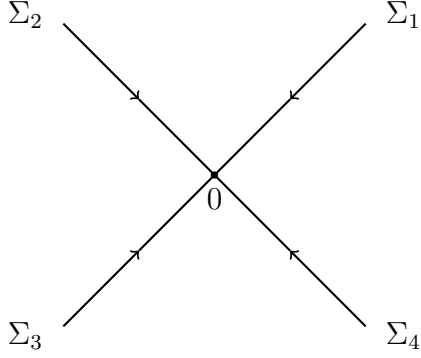
Let $ \Sigma $ denote the contour $\left\{k=\alpha k_0 e^{\pm \frac{\pi i}{4}}: \alpha \in \mathbb{R} \right\}$. We introduce the scaling operator
\begin{equation}\label{Nf}
\begin{aligned}
N: &\mathscr{L}^2( \hat{\Sigma}) \rightarrow \mathscr{L}^2( {\Sigma}) \\
& f(k) \mapsto (Nf)(k)=f\left(k_0+k \sqrt{\frac{i}{4t}}\right).
\end{aligned}
\end{equation}
 When the scaling operator acts on the exponential term and
$\delta(k)$, direct calculations indicate that
\begin{equation}
	N(e^{it\theta} \delta(k))=\delta^0 \delta^1(k),
\end{equation}
where
\begin{equation}\label{delta 0}
\begin{aligned}
\delta_0&=(4t)^{-\frac{i \nu}{2}} e^{\chi(k_0)} e^{k_0^2 t}, \\ 
\delta_1&=k^{i \nu} e^{\chi(k_0 + k\sqrt{\frac{i}{4t}}) - \chi(k_0)}e^{-\frac{i k^2}{4}}.
\end{aligned}
\end{equation}
Furthermore, the jump matrix transforms as
\begin{equation}
N J^{(3)}(k)= J^{(3)}(k_0+k\sqrt{\frac{i}{4t}}).
\end{equation}
Under this scaling transformation, the structure of the jump matrix is illustrated in figure \ref{fig:5}. We set
\begin{equation}
	\omega^{(4)}=N \omega^{(3)}.
\end{equation}
By variable transformation, we derive
\begin{equation}\label{C wA VT}
	C_{\omega^{(3)}}= N^{-1} C_{\omega^{(4)}} N,
\end{equation}
where $ C_{\omega^{(4)}} : \mathscr{L}^2(\Sigma^{(4)}) \to
\mathscr{L}^2(\Sigma^{(4)})$ is bounded. Define
\begin{equation}
\begin{aligned}
	L_A&=\left\{k=\sqrt{\frac{4t}{i}}\alpha e^{\frac{3\pi i}{4}}: -\infty < \alpha <+\infty \right\}, \\ 
	L_A^{*}&=\left\{k=\sqrt{\frac{4t}{i}}\alpha e^{\frac{\pi i}{4}}: -\infty < \alpha <+\infty \right\}.
	\end{aligned}
\end{equation}
For $k \in L_A,$ we have
\begin{equation}
	\omega^{(4)}=\omega^{(4)}_{+}=  \begin{pmatrix} 0 &  0  \\N(R_1(k)\delta^{-2}(k) e^{-2it\theta})  & 0\end{pmatrix},
\end{equation}
and for $k \in L_A^{*},$
\begin{equation}
	\omega^{(4)}=\omega^{(4)}_{-}=  \begin{pmatrix} 0 &  N(R_2(k)\delta^2(k) e^{2it\theta})   \\	0 & 0\end{pmatrix}.
\end{equation}
\begin{lemma}\label{N lesssim} As $t \to \infty$, the following estimates hold:
\begin{equation}
\begin{aligned}
&\left |  N(R(k))(\delta^1(k))^{-2}-R(k_0 \pm)k^{-2i \nu} e^{\frac{i k^2}{2}}  \right |
	\lesssim \frac{\log t}{\sqrt{t}}, \quad & k \in L, \\
	&\left |  N(R^{\prime}(k))(\delta^1(k))^2-R^{\prime}(k_0 \pm)k^{2i \nu} e^{-\frac{i k^2}{2}}  \right |
	\lesssim \frac{\log t}{\sqrt{t}}, \quad & k \in L^{*}, 
\end{aligned}
\end{equation}
where
\begin{align*}
	& R(k_0+)=  \lim_{\re k>k_0 } R(k)=\gamma_1(k_0),
	&& R(k_0-)=\lim_{\re k<k_0} R(k)=-\frac{\gamma_1(k_0)}{1+\gamma_1(k_0)\gamma_2(k_0)}, \\
	& R^{\prime}(k_0+)=\lim_{\re k>k_0}R^{\prime}(k)= \gamma_2(k_0),
	&& R^{\prime}(k_0-)=\lim_{\re k<k_0}R^{\prime}(k)=
	-\frac{\gamma_2(k_0)}{1+\gamma_1(k_0)\gamma_2(k_0)}.
\end{align*}
\end{lemma}
\begin{proof}
For $k \in \Sigma_4,$ let $\beta$ be a fixed constant satisfying $0<2\beta <1.$ Then we obtain
\begin{equation}
\begin{aligned}
&\left|[R]( \sqrt{i/4t}k + k_0 ) (\delta^1)^{-2} -\gamma_1(k_0) k^{-2i \nu} e^{\frac{ik^2}{2}} \right| \\ 
= & \left|[R]( \sqrt{i/4t}k + k_0 ) k^{-2i \nu} e^{-2 \left [
\chi(\sqrt{i/4t} k + k_0)-\chi(k_0)\right ] } e^{\frac{i k^2}{2}} -\gamma_1(k_0) k^{-2i \nu} e^{\frac{ik^2}{2}} \right| \\ 
= &  \left| e^{\frac{i \beta k^2}{2} }\left([R]( \sqrt{i/4t}k + k_0 ) - \gamma_1(k_0)\right) k^{-2i\nu} e^{-2 \left [
\chi(\sqrt{i/4t} k + k_0)-\chi(k_0)\right ] } e^{\frac{i k^2(1-2\beta)}{2}} \right. \\ 
& \left.  +  e^{ \frac{i \beta k^2}{2} } \left( e^{-2 \left [
\chi(\sqrt{i/4t} k + k_0)-\chi(k_0)\right ] }-1 \right) \gamma_1(k_0)   k^{-2i\nu} e^{\frac{i k^2(1-2\beta)}{2}} \right|.
\end{aligned}
\end{equation}
For $k \in \Sigma_4 $ with $ k=u e^{-\frac{ \pi i}{4}}$, we have $\re (i k^2)=  u^2 \ge 0.$ Thus,
\begin{equation}
|e^{- \frac{i \beta k^2}{2} }| \le 1, \quad |e^{\frac{i k^2(2\beta-1)}{2}}| \le 1.
\end{equation}
Recall that $R(k)$ admits an explicit rational representation,
\begin{equation}
R(k)=\frac{\sum_{j=0}^{m}\mu_j(k-k_0)^j}{(k-i)^{m+5}}. 
\end{equation}
We arrive at the following estimates:
\begin{equation}
\begin{aligned}
& |R(k_0)|= O(|k-i|^{-m-5}) \le c, \quad |[R](\sqrt{i/4t} k +k_0)|=O(|k-i|^{j-m-5}) \le c, \\ 
& \nu = \frac{1}{2 \pi} \log (1+\gamma_1(k_0) \gamma_2(k_0)) \le \frac{1}{2 \pi} \log (1+ \| \gamma_1(k) \gamma_2(k)\|) < \infty.
\end{aligned}
\end{equation}
Therefore,
\begin{equation}
\begin{aligned}
& |k^{2 i \nu}| =|e^{2i \nu (\log|k|-\frac{\pi i}{4})}|=e^{\frac{ \pi \nu}{2}} \le c , \\  
& N \delta = (\sqrt{4t/i})^{-i \nu/2} k^{i \nu} e^{\chi(k_0)}e^{\chi(\sqrt{i/4t} k + k_0) - \chi(k_0)}.
\end{aligned}
\end{equation}
It follows that
\begin{equation}
e^{2 \left[\chi(\sqrt{i/4t} k + k_0) - \chi(k_0) \right] }= 
(N \delta)^{2} k^{-2i \nu} (\sqrt{4t/i})^{i \nu}e^{-2 \chi(k_0)}.
\end{equation}
Since $|(\sqrt{4t/i})^{i \nu} e^{-2 \chi(k_0)} |=1$, $(N \delta)^{2}$ is bounded, and $|k^{-2i \nu}|=e^{-\frac{\pi \nu}{2}}<1$, we conclude
\begin{equation}
\left|e^{2 \left[\chi(\sqrt{i/4t} k + k_0) - \chi(k_0) \right] } \right| \le c.
\end{equation}
Furthermore, we obtain the estimate
\begin{equation}
\begin{aligned}
& \left| e^{- \frac{i \beta k^2}{2} }\left([R]( \sqrt{i/4t}k + k_0 ) - \gamma_1(k_0)\right) \right| \\ 
\le & |k| e^ {-\frac{\beta \re(i k^2)}{2}} \| [R]^{\prime} \|_{\mathscr{L}^{\infty}} \sqrt{i/4t} \le ct^{-1/2},
\end{aligned}
\end{equation}
and
\begin{equation}
\begin{aligned}
& \left| e^{ -\frac{i \beta k^2}{2} } \left( e^{2 \left [
\chi(\sqrt{i/4t} k + k_0)-\chi(k_0)\right ] }-1 \right) \right| \\ 
= & |e^{ -\frac{i \beta k^2}{2} }| \left| 2 \int_0 ^1 e^{2s \left[\chi(\sqrt{i/4t} k + k_0) - \chi(k_0)  \right]} \mathrm{d} s \left[\chi(\sqrt{i/4t} k + k_0) - \chi(k_0) \right]\right| \\ 
\le & \sup_{0 \le s \le 1} \left| e^{2s \left[ \chi(\sqrt{i/4t} k + k_0) - \chi(k_0)  \right]} \right| \left| 2e^{-\frac{i \beta k^2}{2}} \left[ \chi(\sqrt{i/4t} k + k_0) - \chi(k_0)   \right] \right| \\ 
\le & c \left| e^{ -\frac{i \beta k^2}{2} } \left[ \chi(\sqrt{i/4t} k + k_0) - \chi(k_0)   \right]\right|,
\end{aligned}
\end{equation}
for $k \in \Sigma_4,$
\begin{equation}\label{the second part}
\begin{aligned}
& \left| e^{ -\frac{i \beta k^2}{2} } \left[ \chi(\sqrt{i/4t} k + k_0) - \chi(k_0)   \right] \right| \\ 
= & \left| e^{ -\frac{i \beta k^2}{2} }  \frac{1}{2 \pi} \int^{k_0} _{-\infty} \left[ \log(\sqrt{i/4t} k + k_0 -s) - \log(k_0-s) \right]\mathrm{d} \log(1+\gamma_1 \gamma_2)\right| \\ 
=& \left| \frac{e^{ -\frac{i \beta k^2}{2} }}{2 \pi} \int^{k_0} _{-\infty} \log \frac{\sqrt{i/4t} k + k_0 -s}{k_0-s}\mathrm{d} \log(1+\gamma_1 \gamma_2)\right| \\ 
=& \left| \frac{e^{ -\frac{i \beta k^2}{2} }}{2 \pi} \int^{\infty} _{0} \log \frac{\sqrt{i/4t} k + s}{s}\mathrm{d} \log(1+\gamma_1 (k_0 -s) \gamma_2 (k_0 -s))\right| \\  
\le & \left| \frac{e^{ -\frac{i \beta k^2}{2} }}{2 \pi} \int^{\infty} _{1} \log\left(1+\frac{k}{\sqrt{4t/i}s}\right) g(s)\mathrm{d} s\right| + 
\left| \frac{e^{ -\frac{i \beta k^2}{2} }}{2 \pi} g(0) \int^{1} _{0} \log \frac{\sqrt{i/4t} k + s}{s}\mathrm{d} s\right| \\ 
& + \left| \frac{e^{ -\frac{i \beta k^2}{2} }}{2 \pi} \int^{1} _{0} \log\left(1+\frac{k}{\sqrt{4t/i}s}\right) [g(s)-g(0)]\mathrm{d} s\right| ,
\end{aligned}
\end{equation}
where 
\begin{equation}
g(s)=\partial_s \log (1+\gamma_1(k_0 -s) \gamma_2 (k_0 -s)) .
\end{equation}
We note the inequality
\begin{equation}
\left| \log(1+m) \right| \le |m|, \quad \re m \ge 0,
\end{equation}
which implies
\begin{equation}\label{inequality 1}
\left| \log \left( 1+\frac{k}{\sqrt{4t/i}s} \right) \right| \le \left| \frac{k}{\sqrt{4t/i}s} \right|.
\end{equation}
Using \eqref{inequality 1}, the sum of the first and third terms of \eqref{the second part} can be estimated as follows:
\begin{equation}\label{part-13}
\begin{aligned}
 &\left| \frac{e^{ -\frac{i \beta k^2}{2} }}{2 \pi} \int^{\infty} _{1} \log\left(1+\frac{k}{\sqrt{4t/i}s}\right) g(s)\mathrm{d} s\right| + \left| \frac{e^{ -\frac{i \beta k^2}{2} }}{2 \pi} \int^{1} _{0} \log\left(1+\frac{k}{\sqrt{4t/i}s}\right) [g(s)-g(0)]\mathrm{d} s\right|  \\ 
 \le & \left| \frac{e^{ -\frac{i \beta k^2}{2} }}{2 \pi} \right| \int^{\infty} _{1} \left|\frac{k}{\sqrt{4t/i}s}\right| |g(s)|\mathrm{d} s  
 + \left| \frac{e^{ -\frac{i \beta k^2}{2} }}{2 \pi} \right| \int^{1} _{0} \left|\frac{k}{\sqrt{4t/i}}\right| \left| \frac{g(s)-g(0)}{s} \right|\mathrm{d} s \\ 
 \le & \left| \frac{k e^{ -\frac{i \beta k^2}{2} }}{2 \pi} 
 \left| \frac{1}{\sqrt{4t/i}} \right|\int^{\infty}_1 \frac{g}{s} \mathrm{d} s  \right|  
 + \left| \frac{k e^{ -\frac{i \beta k^2}{2} }}{2 \pi} \right| \frac{1}{\sqrt{4t/i}}
 \int_0^1 \| g^{\prime}(s) \|_{\mathscr{L}^{\infty}} \mathrm{d} s \\ 
 \le & c t^{-1/2}.
\end{aligned}
\end{equation}
Next, we estimate the second term in  \eqref{the second part}:
\begin{equation}
\begin{aligned}
&\int^{1} _{0} \log \frac{\sqrt{i/4t} k + s}{s}\mathrm{d} s \\ 
=& \int^{1} _{0} \left[\log (\sqrt{i/4t} k + s)- \log s\right] \mathrm{d} s \\ 
=& \left \{  (\sqrt{i/4t} k + s) \log (\sqrt{i/4t} k + s)-(\sqrt{i/4t} k + s) \right \} \bigg| _0^1 - \left\{ s \log s -s \right\} \bigg| _0^1  \\ 
=& (\sqrt{i/4t} k + 1) \log (\sqrt{i/4t} k + 1) - \sqrt{i/4t} k \log \sqrt{i/4t} k.
\end{aligned}
\end{equation}
As a result, for $k \in \Sigma_4,$
\begin{equation}\label{part-2}
\begin{aligned}
& \left| e^{-\frac{i \beta k^2}{2}} \int^1_0 \log \frac{\sqrt{i/4t} k + s}{s}\mathrm{d} s \right| \\ 
\le & \left| \sqrt{i/4t} \right| k e^{-\frac{i \beta k^2}{2}} \left| (1+ \sqrt{i/4t}k) \right| + \left|k e^{-\frac{i \beta k^2}{2}} \right| \left| \frac{\log \sqrt{i/4t}}{\sqrt{4t/i}} \right| + \left| \frac{e^{-\frac{i \beta k^2}{2}} k \log k  }{\sqrt{4t/i}} \right|  \\ 
\le & c t^{-1/2} (1+ \log t) \le ct^{-1/2} \log t, \quad t \ge 2.
\end{aligned}
\end{equation}
Combining \eqref{part-13} and \eqref{part-2}, we obtain
\begin{equation}
\| [R_1]( \sqrt{i/4t}k + k_0 ) (\delta^1)^{-2} -\gamma_1(k_0) k^{-2i \nu} e^{\frac{ik^2}{2}}  \|_{\mathscr{L}^1 \cap \mathscr{L}^2 \cap \mathscr{L}^{\infty}(\Sigma_3)} \le c t^{-1/2} \log t.
\end{equation}
Analogous estimates hold on $\Sigma_1, \Sigma_2$ and $\Sigma_3$.
\end{proof}
Until now, we can structure $J^{A}(y,t,k)$ with controlled error terms. Let
$J^{A}= (I-\omega^A_{-})^{-1} (I+\omega^A_{+})$, where
\begin{equation}
\begin{aligned}
    \omega^A= \omega^A_{-}= \left\{\begin{aligned}
    &	\begin{pmatrix}
    	0     &
    	\gamma_2(k_0)(\delta^0)^2k^{2i \nu} e^{-\frac{i k^2}{2}} \\
    	0    &  0 \\
        \end{pmatrix}, \quad & k \in \Sigma_1,   \\
    &\begin{pmatrix} 0   &
    	-\frac{ \gamma_2(k_0)}{1+\gamma_1(k_0)\gamma_2(k_0)} (\delta^0)^2k^{2i \nu} e^{-\frac{i k^2}{2}} \\ 0 &  0 \\
    \end{pmatrix} , \quad  & k \in \Sigma_3, \\
    \end{aligned}\right.
\end{aligned}
\end{equation}
\begin{equation}
\begin{aligned}
    \omega^A= \omega^A_{+}= \left\{\begin{aligned}
    &	\begin{pmatrix}
    	0     &
    	0\\
    	-\frac{\gamma_1(k_0)}{1+\gamma_1(k_0)\gamma_2(k_0)}(\delta^0)^{-2}k^{-2i \nu} e^{\frac{i k^2}{2}}    &  0 \\
        \end{pmatrix},  \quad & k \in \Sigma_2,   \\
   &\begin{pmatrix} 0   &
    	0\\ \gamma_1(k_0)(\delta^0)^{-2} k^{-2i \nu} e^{\frac{i k^2}{2}}&  0 \\
    \end{pmatrix} ,   \quad & k \in \Sigma_4. \\
    \end{aligned}\right.
\end{aligned}
\end{equation}

It follows from Lemma \ref{N lesssim} and (3.78) in \cite{DeiftZ} that
\begin{equation}
	\left \| \frac{\omega^{(4)}-\omega^{A}}{(N k)^j}  \right \| _{\mathscr{L}^1(\Sigma_A) \cap \mathscr{L}^2(\Sigma_A) \cap \mathscr{L}^{\infty}(\Sigma_A)}  \lesssim \frac{\log t}{\sqrt{t}}, \quad  j=1,2.
\end{equation}
\begin{lemma}\label{Lemma C w}
As $t \to \infty$, for $j=1,2,$
\begin{equation}\label{C decompose}
	\begin{aligned} 
	& \int_{\Sigma^{(3)}} \frac{((1-C_{\omega^{(3)}})^{-1}I)(\xi) \omega^{(3)}(\xi)}{\xi ^j} \mathrm{d} \xi \\
	= & \sqrt{\frac{i}{4t}}  \int_{\Sigma^{(4)}}\frac{\left((1-C_{\omega^{A}})^{-1}I\right)(\xi) \omega^{A}(\xi)}{(N \xi ) ^j} \mathrm{d} \xi  + O\left(\frac{ \log t}{t}\right). \end{aligned}
\end{equation}
\end{lemma}
\begin{proof}
Note that
\begin{equation*}
	\begin{aligned}
	&	\frac{\left((1-C_{\omega^{(4)}})^{-1}I\right)\omega^{(4)}}{(N k)^j}-\frac{\left((1-C_{\omega^{A}})^{-1}I\right)\omega^{A}}{(N k)^j}\\
	=& \frac{(1-C_{\omega^{A}})^{-1}(C_{\omega^{(4)}}-C_{\omega^{A}})I \omega^{A}} {(N k)^j} +
	\frac{\left[	(1-C_{\omega^{A}})^{-1}(C_{\omega^{(4)}}-C_{\omega^{A}})(1-C_{\omega^{A}})^{-1}C_{\omega^{A}}\right] I \omega^{(4)}} {(N k)^j} \\
	& + \frac{\omega^{(4)}-\omega^{A}}{(N k)^j} + (1-C_{\omega^{(4)}})^{-1} C_{\omega^{(4)}} I  \frac{\omega^{(4)}-\omega^{A}}{(N k)^j}.
	\end{aligned}
\end{equation*}
Similar to Lemma \ref{N lesssim}, we obtain
\begin{equation}
	\int_{\Sigma^{(4)}} \frac{\left((1-C_{\omega^{(4)}})^{-1}I\right)(\xi) \omega^{(4)}(\xi)}{(N \xi)^j} \mathrm{d} \xi
	= \int_{\Sigma^{(4)}} \frac{\left((1-C_{\omega^{A}})^{-1}I\right)(\xi) \omega^{A}(\xi)}{(N \xi)^j} \mathrm{d} \xi+ O\left(\frac{ \log t}{\sqrt{t}}\right).
\end{equation}
By a simple change of variable for $j = 1, 2,$ we get
\begin{equation}
	\begin{aligned}
	&	\int_{\Sigma^{(3)}}\frac{((1-C_{\omega^{(3)}})^{-1}I)(\xi) \omega^{(3)}(\xi)}{\xi^j} \mathrm{d} \xi \\
	=& \int_{\Sigma^{(3)}}\frac{N^{-1}(1-C_{\omega^{(4)}})^{-1}N I(\xi)\omega^{(4)}(\xi)}{\xi^j} \mathrm{d} \xi \\
	=& \int_{\Sigma^{(3)}} \frac{(1-C_{\omega^{(4)}})^{-1}I\left((\xi-k_0) \sqrt{\frac{4t}{i}}\right) (N  \omega^{(4)})\left((\xi-k_0) \sqrt{\frac{4t}{i}}\right)}{\xi ^j} \mathrm{d} \xi \\
	=&  \sqrt{\frac{i}{4t}} \int_{\Sigma^{(4)}} \frac{\left((1-C_{\omega^{(4)}})^{-1}I\right)(\xi) \omega^{(4)}(\xi) }{(N\xi)^j}  \mathrm{d} \xi \\
	=&  \sqrt{\frac{i}{4t}} \int_{\Sigma^{(4)}} \frac{((1-C_{\omega^{A}})^{-1}I)(\xi) \omega^{A}(\xi) }{(N\xi)^j}  \mathrm{d} \xi  + O\left(\frac{ \log t}{t}\right).
	\end{aligned}
\end{equation}
\end{proof}

Next, we construct a RH problem based on \eqref{C decompose}. For $k \in \mathbb{C} \setminus \Sigma^{(4)}$, set
\begin{equation}
	M^{A}(y,t,k)=I+ \frac{1}{2 \pi i} \int_{\Sigma^{(4)}} \frac{((1-C_{\omega^{A}})^{-1}I)(\xi)\omega^{A}(\xi )}{\xi -k} \mathrm{d} \xi .
\end{equation}
Then $	M^{A}(y,t,k) $ is the solution of the RH problem
\begin{equation}
	\begin{aligned}\label{RHP M A0}
	\left\{\begin{array}{ll}  	M_+^{A}(y,t,k)= M_-^{A}(y,t,k) J^{A}(y,t,k), & \quad k \in \Sigma^{(4)}, \\
	M^{A}(y,t,k) \to I, & \quad k \to \infty. \end{array}\right.
	\end{aligned}
\end{equation}
Specifically,
\begin{equation}\label{M A0 expansion}
	M^{A}(y,t,k) =I+ \frac{	M_1^{A}(y,t)}{k} + O(k^{-2}),  \quad  k \to \infty.
\end{equation}
Combining the above results, a direct computation yields:
\begin{equation}\label{M A0 + O}
	\begin{aligned}
	&\sqrt{\frac{i}{4t}} \int_{\Sigma^{(4)}} \frac{\left((1-C_{\omega^{A}})^{-1}I\right)(\xi)\omega^{A}(\xi )}{ N \xi}
	\frac{\mathrm{d} \xi }{2 \pi i} \\
	=&  \int_{\Sigma^{(4)}} \frac{\left((1-C_{\omega^{A}})^{-1}I\right)(\xi)\omega^{A}(\xi )}{ \xi+ k_0 \sqrt{\frac{4t}{i}}}
	\frac{\mathrm{d} \xi }{2 \pi i} \\
	=& M^{A}(-k_0 \sqrt{\frac{4t}{i}})-I \\
	=&- \frac{ M_1^{A}(y,t)}{k_0\sqrt{\frac{4t}{i}}} +O(t^{-1}).
	\end{aligned}
\end{equation}
Furthermore, we have
\begin{equation}\label{-M A0 + O}
	\begin{aligned}
	&\sqrt{\frac{i}{4t}} \int_{\Sigma^{(4)}} \frac{\left((1-C_{\omega^{A}})^{-1}I\right)(\xi)\omega^{A}(\xi )}{( N \xi)^2} \frac{\mathrm{d} \xi }{2 \pi i} \\
	=& \sqrt{\frac{4t}{i}}  \int_{\Sigma^{(4)}} \frac{\left((1-C_{\omega^{A}})^{-1}I\right)(\xi)\omega^{A}(\xi )}{ (\xi  +k_0 \sqrt{\frac{4t}{i}})^2}
	\frac{\mathrm{d} \xi }{2 \pi i} \\
	=&  - \sqrt{\frac{4t}{i}}   \frac{\mathrm{d}{M^{A}(k)}}{\mathrm{d} k} \bigg|_{k = k_0 \sqrt{\frac{4t}{i}}}  \\
	=&  \frac{ M_1^{A}(y,t)}{k_0^2\sqrt{\frac{4t}{i}}} +O(t^{-1}).
	\end{aligned}
\end{equation}
From \eqref{tilde u4}, \eqref{tilde v4}, \eqref{M A0 + O}, \eqref{-M A0 + O} and Lemma \ref{Lemma C w}, we deduce
\begin{equation}\label{tilde u6}
\begin{aligned}
	\tilde{u}(y,t)=& \left[I - \frac{M_1^{A}(y,t)}{k_0\sqrt{\frac{4t}{i}}} +O\left(\frac{\log t}{t}\right)	\right]_{22}  \times  \left[ \frac{M_1^{A}(y,t)}{k_0^2\sqrt{\frac{4t}{i} }} + O\left(\frac{\log t}{t}\right)			\right]_{12} \\
	&-\left[I- \frac{M_1^{A}(y,t)}{k_0\sqrt{\frac{4t}{i}}} +O\left(\frac{\log t}{t}\right)	\right]_{12}  \times  \left[\frac{M_1^{A}(y,t)}{k_0^2\sqrt{\frac{4t}{i} }}  + O\left(\frac{\log t}{t}\right)			\right]_{22} \\
	& = \frac{\left[M_1^{A}(y,t)\right]_{12}}{k_0^2\sqrt{\frac{4t}{i} }} + O\left(\frac{\log t}{t}\right)	,
\end{aligned}
\end{equation}

\begin{equation}
	\begin{aligned}\label{tilde v6}
\tilde{v}(y,t)	=& -\left[ I - \frac{M_1^{A}(y,t)}{k_0\sqrt{\frac{4t}{i}}} +O\left(\frac{\log t}{t}\right) \right]_{11} \left[ \frac{M_1^{A}(y,t)}{k_0^2\sqrt{\frac{4t}{i} }} + O\left(\frac{\log t}{t} \right) \right]_{21} \\
	&+\left[ I - \frac{M_1^{A}(y,t)}{k_0\sqrt{\frac{4t}{i}}} +O\left(\frac{\log t}{t}\right) \right]_{21} \left[\frac{M_1^{A}(y,t)}{k_0^2\sqrt{\frac{4t}{i} }} + O\left(\frac{\log t}{t} \right)\right]_{11} \\ 
	&=- \frac{\left[M_1^{A}(y,t)\right]_{21}}{k_0^2\sqrt{\frac{4t}{i} }} + O\left(\frac{\log t}{t}\right).
	\end{aligned}
\end{equation}
\subsection{Solving the model problem}
To obtain $(M_1^{A})_{12}$ explicitly from \eqref{tilde u6}, we consider the following transformation:
\begin{equation}
	\phi(y,t,k)=(\delta^0)^{-\sigma_3} M^{A} (y,t,k) (\delta^0)^{\sigma_3} k^{-i \nu \sigma_3}e^{\frac{1}{4}i k^2 \sigma_3},
\end{equation}
which implies
\begin{equation}
	\phi_+(k)= \phi_-(k) v(k_0), \quad  v(k_0)=e^{-\frac{1}{4}i k ^2 \hat{\sigma}_3} k^{i \nu \hat{\sigma}_3} (\delta^0)^{-\hat{\sigma}_3} J^{A}(y,t,k).
\end{equation}
 Since the jump matrix is constant along each ray $ \Sigma_1, \Sigma_2, \Sigma_3$ and $ \Sigma_4,$ we have
\begin{equation}
	\frac{\mathrm{d} \phi_+(k)}{\mathrm{d} k}-\frac{i}{2} k \sigma_3 \phi_+(k)
	=\left(\frac{\mathrm{d} \phi_- (k)}{\mathrm{d} k} - \frac{i}{2} k \sigma_3 \phi_- (k)\right) v(k_0).
\end{equation}
It follows that $ \left(\frac{\mathrm{d} \phi }{\mathrm{d} k} - \frac{i}{2} k \sigma_3 \phi \right) \phi^{-1}$ has no jump discontinuity on $\Sigma$. Additionally, from the relation between $\phi(k)$ and $M^{A}(k)$, we have
\begin{equation}
	\left(\frac{\mathrm{d} \phi (k)}{\mathrm{d} k} - \frac{i}{2} k \sigma_3 \phi (k)\right)
	\phi^{-1}(k) = O(k^{-1}) + \frac{i}{2} (\delta^0)^{-\sigma_3} [ M_1^{A},\sigma_3] (\delta^0)^{\sigma_3}.
\end{equation}
By Liouville's Theorem, we conclude that
\begin{equation}\label{beta phi}
	\frac{\mathrm{d} \phi (k)}{\mathrm{d} k} -\frac{i}{2} k \sigma_3 \phi(k)= \beta \phi(k),
\end{equation}
where
\begin{equation}
	\beta=\frac{i}{2} (\delta^0)^{-\sigma_3} [ M_1^{A},\sigma_3] (\delta^0)^{\sigma_3}.
\end{equation}
Particularly,
\begin{equation}\label{beta 12}
	\beta_{12}=- i (M_1^{A})_{12} (\delta^0)^{-2}, \quad \beta_{21}= i (M_1^{A})_{21} (\delta^0)^2.
\end{equation}
From \eqref{beta phi} and its differential, we obtain
\begin{align}
	&	\frac{\mathrm{d}^2 \phi_{11}(k)}{\mathrm{d} k^2} +\left(-\frac{i}{2}+\frac{1}{4}k^2- \beta_{12} \beta_{21}\right) \phi_{11}(k)=0, \label{phi 11} \\
	&   \beta_{12} \phi_{21} (k)=  \frac{\mathrm{d} \phi_{11}(k)}{\mathrm{d} k} -\frac{i}{2} k \phi_{11}(k), \label{beta12 phi12}\\
	&	\frac{\mathrm{d}^2 \phi_{22}(k)}{\mathrm{d} k^2} +\left(\frac{i}{2}+\frac{1}{4}k^2- \beta_{12} \beta_{21}\right) \phi_{22}(k)=0, \label{phi 22} \\
	&   \beta_{21} \phi_{12} (k)=  \frac{\mathrm{d} \phi_{22}(k)}{\mathrm{d} k} +\frac{i}{2} k \phi_{22}(k).
	\end{align}
Set $ \phi_{11}(k)=g(e^{-\frac{ \pi i }{4}} k)$  and $\phi_{22}(k)=g(e^{-\frac{3 \pi i }{4}} k) $. Then \eqref{phi 11} and \eqref{phi 22} turn to the Weber’s equation
\begin{equation}\label{Weber}
	\frac{\mathrm{d}^2 g(\zeta )}{\mathrm{d} \zeta ^2}+\left(a+\frac{1}{2}-\frac{\zeta^2}{4}\right) g(\zeta )=0.
\end{equation}
As is well-known, \eqref{Weber} is a second-order ordinary differential equation, which has two linearly independent solutions
$ D_a(\zeta)$ and $D_a(-\zeta)$. Thus, there are constants $c_1$ and $c_2$ such that
\begin{equation*}
	g(\zeta)=c_1D_a(\zeta)+c_2D_{a}(-\zeta),
\end{equation*}
where $D_a(\cdot) $ represents the standard parabolic-cylinder function and satisfies
\begin{align}
	& \frac{\mathrm{d} D_a(\zeta)}{\mathrm{d} \zeta} + \frac{\zeta}{2}  D_a(\zeta) -a  D_{a-1}(\zeta)=0, \label{D a-1}\\
	& D_a(\pm \zeta)= \frac{\Gamma(a+1)e^{\frac{a \pi i}{2}}}{\sqrt{2 \pi }} D_{-a-1}(\pm i\zeta)+
	\frac{\Gamma(a+1)e^{-\frac{a \pi i}{2}}}{\sqrt{2 \pi }} D_{-a-1}(\mp i\zeta). \label{Da pm}
\end{align}
As $\zeta \to \infty, $ we have \cite{ModelRHP}
\begin{equation}\label{Da}
	D_a(\zeta)= \left\{\begin{array}{ll} \zeta^a e^{-\frac{\zeta^2}{4}} (1+O(\zeta^{-2})),  & |\arg \zeta| < \frac{3 \pi }{4},	\\
	\zeta^a e^{-\frac{\zeta^2}{4}} (1+O(\zeta^{-2})) - \frac{\sqrt{2 \pi }}{\Gamma (-a)} e^{a \pi i+ \frac{\zeta^2}{4}}
	\zeta ^{-a-1} (1+O(\zeta^{-2})),   & \frac{ \pi }{4}< \arg \zeta < \frac{5 \pi }{4}, \\
	\zeta^a e^{-\frac{\zeta^2}{4}} (1+O(\zeta^{-2})) - \frac{\sqrt{2 \pi }}{\Gamma (-a)} e^{-a \pi i+ \frac{\zeta^2}{4}}
	\zeta ^{-a-1} (1+O(\zeta^{-2})),   & -\frac{5 \pi }{4}< \arg \zeta < - \frac{ \pi }{4},
	\end{array}\right.
\end{equation}
where $ \Gamma(\cdot)$ is the Gamma function. Choose $a= -i \beta_{12} \beta_{21},$
\begin{align}
	& \phi_{11}(k)=c_1 D_a (e^{-\frac{ \pi i}{4}}k)+c_2 D_a(e^{\frac{3\pi i}{4}}k), \label{phi11 expression} \\
	& \phi_{22}(k)=c_3 D_{-a}(e^{- \frac{ 3\pi i}{4}} k)+c_4 D_{-a}(e^{\frac{ \pi i}{4}}k).
\end{align}
As $\arg k \in (\frac{\pi}{4}, \frac{3\pi}{4})$ and $k \to \infty,$ we achieve
\begin{equation}\label{phi 11 and phi 22}
	\phi_{11}(k) k^{i\nu } e^{-\frac{i k^2}{4}} \to 1, \quad 
	\phi_{22}(k) k^{-i\nu } e^{\frac{i k^2}{4}} \to 1.
\end{equation}
Set $k=\tau e^{\frac{\pi i}{4}}$. We obtain from \eqref{phi 11 and phi 22} that
\begin{equation}\label{phi 11 expression2}
	\begin{aligned}
	\phi_{11}(k)=\phi_{11}(\tau e^{\frac{\pi i}{4}})= &[1+O(\tau ^{-1})](\tau e^{\frac{\pi i}{4}})^{-i \nu } e^{\frac{i}{4}(\tau e^{\frac{\pi i}{4}})^2} \\
	=& [1+O(\tau ^{-1})] \tau^{-i \nu } e^{\frac{\pi \nu}{4} } e^{-\frac{\tau^2}{4}}.
	\end{aligned}
\end{equation}
Since $\arg k \in (\frac{\pi}{4}, \frac{3\pi}{4})$, we have $\arg \tau \in (0, \frac{\pi}{2}) $. From the first equation
of $D_a(\zeta)$ in \eqref{Da}, we obtain
\begin{equation}\label{Da result}
	\begin{aligned}
	&D_a(e^{-\frac{\pi i}{4}} k) =D_a(\tau) =\tau^a e^{-\frac{\tau^2}{4}} [1+O(\tau^{-2})],\\
	&D_a(e^{\frac{3\pi i}{4}} k) = D_a(-\tau) =(-\tau)^a e^{-\frac{\tau^2}{4}} [1+O(\tau^{-2})].		
	\end{aligned}
\end{equation}
Substituting \eqref{Da result} into equation \eqref{phi11 expression} and comparing with \eqref{phi 11 expression2}, we conclude that
\begin{equation}\label{parameter}
	c_1=e^{\frac{\pi \nu}{4}}, \quad  c_2=0, \quad  a=-i\nu .
\end{equation}
Thus,
\begin{equation}
	\phi_{11}(k)= e^{\frac{\pi \nu }{4}} D_a(e^{-\frac{ \pi i}{4}} k),
\end{equation}
then $\nu = \beta_{12} \beta_{21}$. Similarly, setting $ k= \kappa  e^{\frac{3\pi i}{4}}$, we obtain
\begin{equation}
	\phi_{22}(k)= e^{-\frac{3\pi \nu }{4}} D_{-a}(e^{-\frac{ 3\pi i}{4}} k).
\end{equation}
Combining this with \eqref{beta12 phi12} and \eqref{D a-1}, we deduce
\begin{equation}
	\beta_{12} \phi_{21}(k)= e^{\frac{\pi \nu }{4}} e^{-\frac{\pi i}{4}} a D_{a-1}(e^{-\frac{ \pi i}{4}} k).
\end{equation}
For $\arg k \in (-\pi, -\frac{3\pi}{4})\cup (\frac{3\pi}{4}, \pi)$ and $k \to \infty,$ in the same manner, we arrive at
\begin{align}
	& \phi_{11}(k)= e^{-\frac{3 \pi \nu }{4}} D_a(e^{\frac{3 \pi i}{4}} k), \quad
	\phi_{22}(k)= e^{\frac{\pi \nu }{4}} D_{-a}(e^{\frac{ \pi i}{4}} k),  \\
	& \beta_{12} \phi_{21}(k)= e^{-\frac{3 \pi \nu }{4}} e^{\frac{3 \pi i}{4}} a D_{a-1}(e^{\frac{ 3 \pi i}{4}} k).
\end{align}
Along the ray $ \arg k =\frac{3 \pi }{4},$ we have
\begin{equation}
	\phi_+(k)=	\phi_-(k) \begin{pmatrix} 1	&  0 \\
	-\frac{\gamma_1(k_0)}{1+\gamma_1(k_0)\gamma_2(k_0)} &  1	 \end{pmatrix}.
\end{equation}
Considering the (2,1) entry of the above RH problem,
\begin{equation}
	e^{\frac{\pi \nu }{4}} e^{-\frac{ \pi i}{4}} a D_{a-1} (e^{-\frac{\pi i}{4}}k)=
	e^{-\frac{3 \pi \nu}{4}} e^{\frac{ 3\pi i}{4}}  a D_{a-1}(e^{\frac{3 \pi i}{4}} k)  - \frac{\gamma_1(k_0)}{1+\gamma_1(k_0)\gamma_2(k_0)} \beta_{12}  e^{\frac{\pi \nu }{4}} D_{-a} (e^{\frac{  \pi i}{4}} k).
\end{equation}
It follows from \eqref{Da pm} that $ D_{-a} (e^{\frac{ \pi i}{4}} k)$ can be decomposed into the sum of
$D_{a-1}(e^{\frac{ 3 \pi i}{4}} k) $ and $D_{a-1}(e^{-\frac{ \pi i}{4}} k) $. Analyzing the coefficients of these two independent functions yields
\begin{equation}
	\beta_{12}= \frac{\sqrt{2 \pi } (1+\gamma_1(k_0)\gamma_2(k_0))  e^{\frac{3\pi i}{4}} e^{\frac{a \pi i}{2}} e^{-\pi \nu} a}{\gamma_1(k_0)  \Gamma(-a+1) }
	=- \frac{\sqrt{2 \pi } (1+\gamma_1(k_0)\gamma_2(k_0)) e^{-\frac{\pi i}{4}} e^{-\frac{a \pi i}{2}} a}{\gamma_1(k_0)   \Gamma(-a+1) }.
\end{equation}
Noting that $a=-i\nu$ in \eqref{parameter} and the property of Gamma function $\Gamma(a+1)=a\Gamma(a)$, we obtain
\begin{equation}\label{beta_12}
\beta_{12}= \frac{\sqrt{2 \pi } (1+\gamma_1(k_0)\gamma_2(k_0)) e^{-\frac{\pi i}{4}} e^{-\frac{{ \pi \nu}}{2}} }{\gamma_1(k_0)  \Gamma(i\nu)}.
\end{equation}
Following the same procedure, we find
\begin{equation}\label{beta_21}
\beta_{21}= -\frac{\sqrt{2 \pi}e^{\frac{\pi i}{4}}e^{-\frac{\pi \nu}{2}}}{\gamma_2(k_0)\Gamma(-i \nu)},
\end{equation}
which satisfies the relation $\nu=\beta_{12}\beta_{21}$ together with \eqref{beta_12}.

Combining \eqref{nu}, \eqref{Chi}, \eqref{delta 0}, \eqref{beta 12}, \eqref{tilde u6}, \eqref{tilde v6}, \eqref{beta_12} and \eqref{beta_21}, we arrive at the final result
\begin{equation}
\begin{aligned}
	\tilde{u}(y,t)=&\frac{\alpha_u(k_0)}{\sqrt{t}}e^{\Theta}+O\left(\frac{\log t}{t} \right),  \\ 
	\tilde{v}(y,t)=& \frac{\alpha_v(k_0)}{\sqrt{t}}e^{\tilde{\Theta}}+O\left(\frac{\log t}{t} \right), 
	\end{aligned} 
\end{equation}
where
\begin{equation*}
\begin{aligned}
\alpha_u(k_0)=&-\frac{1}{2k_0^2}\sqrt{\frac{\nu \gamma_1(k_0)}{i\gamma_2(k_0)} }, \quad \alpha_v(k_0)= \frac{1}{2k_0^2}\sqrt{\frac{\nu \gamma_1(k_0)}{i\gamma_2(k_0)} },\\
	\Theta =& -\frac{\pi i}{4}+\pi \nu +2k_0^2t- i\arg \Gamma (i \nu) -i\nu \log(4t)- i\arg \frac{\gamma_1(k_0)}{1+\gamma_1(k_0)\gamma_2(k_0)} \\ 
& -\frac{1}{\pi}
	\int_{-\infty}^{k_0} \log \left(\frac{1+\gamma_1(\xi)\gamma_2(\xi)}{1+\gamma_1(k_0)\gamma_2(k_0)}\right) \frac{\mathrm{d} \xi}{\xi -k_0}, \\ 
	\tilde{\Theta} =& \frac{\pi i}{4}-\pi \nu -2k_0^2t- i\arg \Gamma (-i \nu) +i\nu \log(4t)- i\arg \gamma_2(k_0) \\ 
& +\frac{1}{\pi}
	\int_{-\infty}^{k_0} \log \left(\frac{1+\gamma_1(\xi)\gamma_2(\xi)}{1+\gamma_1(k_0)\gamma_2(k_0)}\right) \frac{\mathrm{d} \xi}{\xi -k_0},
\end{aligned}
\end{equation*}
and
\begin{equation*}
\begin{aligned}
	k_0=&\frac{y}{2t}, \quad \nu=\frac{1}{2\pi} \log (1+\gamma_1(k_0)\gamma_2(k_0)).
\end{aligned}
\end{equation*}
With these comprehensive results, Theorem \ref{the result} can be established quickly.\\

\noindent{\bf Financial support.} 
This work is supported by the National Natural Science Foundation of China (Grant Nos. 12471234, 12371253, 12471240) and Science Foundation of Henan Academy of Sciences (Grant No. 20252319002).\\

\noindent{\bf Competing interests.} The authors have no conflicts to disclose.\\

\noindent{\bf Data availability statement.} No data was used for the research described in the article.


\begin{thebibliography}{90}
\bibitem{Heisenberg W}
Heisenberg, W. (1928) Zur Theorie des Ferromagnetismus. {\sl Z. Physik} {\bf 49}, 619-636.

\bibitem{Takhtajan-Heisenberg}
Takhtajan, L. A. (1977) Integration of the continuous Heisenberg spin chain through the inverse scattering method. {\sl Phys. Lett. A} {\bf 64}, 235-237.

\bibitem{ZT-1979}
Zakharov, V. E., Takhtajan, L. A. (1979) Equivalence of the nonlinear Schr\"odinger equation and the equation
of a Heisenberg ferromagnet. {\sl Theor. Math. Phys.} {\bf 38}, 17-23.

\bibitem{Darboux-Heisenberg}
Yersultanova, Z. S., Zhassybayeva, M., Yesmakhanova, K., Nugmanova, G., Myrzakulov, R. (2016) 
Darboux transformation and exact solutions of the integrable Heisenberg ferromagnetic equation with self-consistent potentials. 
{\sl Int. J. Geom. Methods Mod. Phys.} {\bf 13}, 1550134.
 
\bibitem{Zha} Zhang, Y., Nie, X. J., Zha, Q. L. (2014) Rogue wave solutions for the Heisenberg ferromagnet equations. {\sl Chin. Phys. Lett.} {\bf 31}, 060201.

\bibitem{CCHF}
Demontis, F., Ortenzi, G., Sommacal, van der Mee, M. C. (2019) The continuous classical Heisenberg ferromagnet equation with in-plane asymptotic conditions. II. IST and closed-form soliton solutions, Ric. Mat. 68 (2019) 163-178.

\bibitem{IST-Gardner}
Gardner, C. S., Greene, J. M., Kruskal, M. D., Miura, R. M. (1967) 
Method for solving the Korteweg-deVries equation. {\sl Phys. Rev. Lett.} {\bf 19}, 1095-1097.

\bibitem{Zakharov}
Zakharov, V. E., Shabat, A. B. (1972) 
Exact theory of two-dimensional self-focusing and one-dimensional self-modulation of waves in nonlinear media. {\sl Sov. Phys. JETP} {\bf 34}, 62-69.

\bibitem{DeiftZ}
Deift, P., Zhou, X. (1993) 
A steepest descent method for oscillatory Riemann-Hilbert problems. Asymptotics for the MKdV equation. {\sl Ann. Math.} {\bf 137}, 295-368.

\bibitem{Andreiev-ELT2016} 
Andreiev, K., Egorova, I., Lange, T. L., Teschl, G. (2016) 
Rarefaction waves of the Korteweg-de Vries equation via nonlinear steepest descent. {\sl J. Differ. Equ.} {\bf 261}, 5371-5410.

\bibitem{Minakov-2011}
Minakov, A. (2011) Long-time behavior of the solution to the mKdV equation with
step-like initial data. {\sl J. Phys. A} {\bf 44}, 085206.

\bibitem{Liunan-mKdV}
Liu, N., Guo, B. L., Wang, D. S., Wang, Y. F. (2019)  Long-time asymptotic behavior for an extended modified Kortweg-de Vries equation. 
{\sl Commun. Math. Sci.} {\bf 17}, 1877-1913.

\bibitem{Deift-P2011} 
Deift, P., Park, J. (2011)  
Long-time asymptotics for solutions of the NLS equation with a delta potential and even initial data.
{\sl Int. Math. Res. Not.} {\bf 2011}, 5505-5624. 

\bibitem{Biondini-M2017} 
Biondini, G., Mantzavinos, D. (2017) 
Long-time asymptotics for the focusing nonlinear Schr\"odinger equation with nonzero boundary conditions at infinity and asymptotic stage of modulational instability.
{\sl Commun. Pure Appl. Math.} {\bf 70}, 2300-2365.

\bibitem{Boutet-KS2022} 
Boutet de Monvel, A., Karpenko, I., Shepelsky, D. (2022) 
The modified Camassa-Holm equation on a nonzero background: large-time asymptotics for the Cauchy problem.
{\sl Pure Appl. Funct. Anal.} {\bf 7}, 887-914.

\bibitem{Boutet-LS2021} 
Boutet de Monvel, A., Lenells, J., Shepelsky, D. (2021) 
The focusing NLS equation with step-like oscillating background: scenarios of long-time asymptotics. 
{\sl Commun. Math. Phys.} {\bf 383}, 893-952.

\bibitem{Boutet-LS2022} 
Boutet de Monvel, A., Lenells, J., Shepelsky, D. (2022) 
The focusing NLS equation with step-like oscillating background: the genus 3 sector.
{\sl Commun. Math. Phys.} {\bf 390}, 1081-1148.

\bibitem{Cheng-VZ1999} 
Cheng, P., Venakides, S., Zhou, X. (1999)   
Long-time asymptotics for the pure radiation solution of the sine-Gordon equation.
{\sl Commun. Partial Differ. Equ.} {\bf 24}, 1195-1262.

\bibitem{Kitaev-V1999} 
Kitaev, A. V.,  Vartanian, A. H. (1999) 
Asymptotics of solutions to the modified nonlinear Schr\"odinger equation: solitons on a nonvanishing continuous background.
{\sl SIAM J. Math. Anal.} {\bf 30}, 787-832.

\bibitem{Boutet-LS2019} 
Boutet de Monvel, A., Lenells, J.,D. Shepelsky, D. (2019)
Long-time asymptotics for the Degasperis-Procesi equation on the half-line.
{\sl Ann. Inst. Fourier} {\bf 69}, 171-230.

\bibitem{Boutet-S2014} 
Boutet de Monvel, A., D. Shepelsky, D. (2014)  
The Ostrovsky-Vakhnenko equation: a Riemann-Hilbert approach.
{\sl C. R. Math. Acad. Sci. Paris} {\bf 352}, 189-195.

\bibitem{Boutet-SZ2016} 
Boutet de Monvel, A., Shepelsky, D., Zielinski, L. (2016) 
A Riemann-Hilbert approach for the Novikov equation.
{\sl SIGMA Symmetry Integrability Geom. Methods Appl.} {\bf 12}, 095.

\bibitem{Geng-L2018} 
Geng, X. G., Liu, H. (2018)  
The nonlinear steepest descent method to long-time asymptotics of the coupled nonlinear Schr\"odinger equation.
{\sl J. Nonlinear Sci.} {\bf 28}, 739-763.

\bibitem{Geng-WC2021} 
Geng, X. G., Wang, K. D., Chen, M. M. (2021) 
Long-time asymptotics for the spin-1 Gross-Pitaevskii equation.
{\sl Commun. Math. Phys.} {\bf 382}, 585-611.

\bibitem{Geng-WC2022} 
Geng, X. G., Wang, K. D., Chen, M. M. (2022)  
The Hermitian symmetric space Fokas-Lenells equation: spectral analysis and long-time asymptotics.
{\sl IMA J. Appl. Math.} {\bf 87}, 852-905.

\bibitem{Liu-GX2018} 
Liu, H.,  Geng, X. G., Xue, B. (2018) 
The Deift-Zhou steepest descent method to long-time asymptotics for the Sasa-Satsuma equation.
{sl J. Differ. Equ.} {\bf 265}, 5984-6008.

\bibitem{Hirota-F}
Huang, L., Xu, J., Fan, E. G. (2015) Long-time asymptotic for the Hirota equation via nonlinear steepest descent method. 
{\sl Nonlinear Anal. Real World Appl.} {\bf 26}, 229-262.

\bibitem{coupled Hirota}
Liu, N., Guo, B.L. (2021) Long-time asymptotics for the initial-boundary value problem of coupled Hirota equation on the half-line.
{\sl Sci. China Math.} {\bf 64}, 81-110.
\bibitem{TianSF}
Tian, S. F., Zhang, T. T. (2018) 
Long-time asymptotic behavior for the Gerdjikov-Ivanov type of derivative nonlinear Schr\"odinger equation with time-periodic boundary condition. {\sl Proc. Amer. Math. Soc.} {\bf 146}, 1713-1729.
\bibitem{GuoBL}
Guo, B. L., Liu, N., (2019) The Gerdjikov‐Ivanov-type derivative nonlinear Schr\"odinger equation: Long‐time dynamics of nonzero boundary conditions. {\sl Math. Methods Appl. Sci.} {\bf 42}, 4839-4861.
\bibitem{ZQMW-GI}
Zheng, F. Z., Qin, Z. Y., Mu, G., Wang, T. Y. (2025)
Long-time asymptotics for the combined nonlinear Schr\"odinger and Gerdjikov-Ivanov equation.
{\sl East Asian J. Appl. Math.} {\bf 15}, 163-184.
\bibitem{discrete-mKdV}
Chen, M., Fan, E. G. (2020) Long-time asymptotic behavior for the discrete defocusing mKdV equation.  
{\sl J. Nonlinear Sci.} {\bf 30}, 953-990.
\bibitem{VanishingLemma}
Ablowitz, M.J., A.S. Fokas, A. S. (2023) 
 Complex Variables: Introduction and Applications, 2nd ed.,
Cambridge University Press, Cambridge.

\bibitem{BC}
Beals, R., Coifman, R. R. (1984) 
Scattering and inverse scattering for first order systems.
{\sl Commun. Pure Appl. Math.} {\bf 37}, 39-90.

\bibitem{ModelRHP}
Whittaker, E. T., Watson, G. N. (1927) A Course of Modern Analysis,
Cambridge University Press, Cambridge.

\end{thebibliography}
\end{document}